\documentclass{llncs}

\usepackage[dvipsnames]{xcolor}
\usepackage{amssymb,amsmath, amsfonts}   
\usepackage{graphics}
\usepackage{wrapfig}
\usepackage{caption}
\usepackage{verbatim}
\usepackage{paralist}
\usepackage{cancel}
\usepackage{xspace}
\usepackage{hyperref}
\usepackage{cleveref}
\usepackage{mathtools}
\usepackage{algorithm}
\usepackage{algorithmicx,algpseudocode}
\usepackage{subcaption}
\usepackage{cite}
\usepackage{makecell}
\usepackage{mathtools}
\usepackage{enumitem}
\usepackage{graphicx}
\usepackage{subcaption}
\usepackage{array}
\usepackage[firstpage]{draftwatermark}     

\definecolor{navy}{RGB}{0,0,128}
\definecolor{dodgerblue}{RGB}{30,144,255}
\newcommand{\relu}{\text{ReLU}\xspace{}}
\newcommand{\mysubsection}[1]{\medskip\noindent\textbf{#1}}
\newcommand{\sat}{\texttt{SAT}}
\newcommand{\unsat}{\texttt{UNSAT}}
\newcommand{\rn}[1]{\mathbb{R}^{#1}}
\newcommand{\nn}{\mathcal{N}}

\newtheorem{Example}{Example}

\usepackage{tikz,ifthen,pgfplots}
\usetikzlibrary{arrows,trees,backgrounds,automata,shapes,decorations,plotmarks,fit,calc,positioning,shadows,chains}
\tikzstyle{every pin edge}=[<-,shorten <=1pt]
\tikzstyle{neuron}=[circle,fill=black!25,minimum size=17pt,inner sep=0pt]
\tikzstyle{input neuron}=[neuron, fill=green!50]
\tikzstyle{output neuron}=[neuron, fill=red!50]
\tikzstyle{hidden neuron}=[neuron, fill=blue!50]
\tikzstyle{merged neuron}=[neuron, fill=orange!50]
\tikzstyle{annot} = [text width=6em, text centered]
\tikzstyle{nnedge} = [-{stealth},shorten >=0.1cm, shorten <=0.05cm,line width=0.8pt,black]
\tikzstyle{proofNode} = [rounded rectangle, fill=red!30]
\tikzstyle{lemmaNode} = [rounded rectangle, fill=dodgerblue!30]
\tikzstyle{lemmaNodeRemoved} = [rounded rectangle, fill=black!10,text=black!75]
\tikzstyle{proofEdge} = [-{stealth},shorten >=0.1cm, shorten <=0.05cm,line width=0.8pt,black]
\usetikzlibrary{calc}

\setlist[itemize]{leftmargin=*}

\title{Proof Minimization in Neural Network Verification}

\author{
 	Omri Isac\inst{1},
 	Idan Refaeli\inst{1},
 	Haoze Wu\inst{2},
 	Clark Barrett\inst{3},
 	Guy Katz\inst{1}
}

\institute{ 
	The Hebrew University of Jerusalem, Israel \and Amherst College, USA \and Stanford University, USA
}

\SetWatermarkAngle{0}
\SetWatermarkText{\raisebox{12.6cm}{		
	\hspace{9cm}
	\href{https://doi.org/10.5281/zenodo.17169557}{\includegraphics[width=20mm,keepaspectratio]{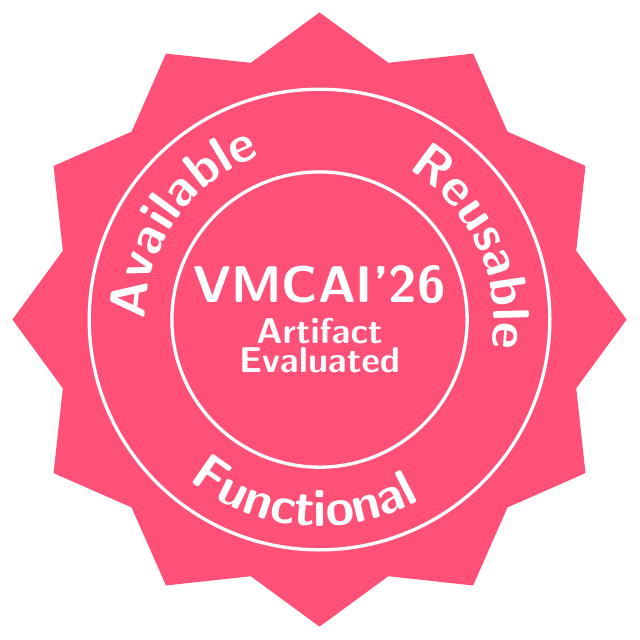}}
} }
\begin{document}

\maketitle
\begin{abstract}
	The widespread adoption of deep neural networks (DNNs) requires efficient techniques for verifying their safety.
	DNN verifiers are complex tools, which might contain bugs that
        could compromise their soundness and undermine the reliability of the verification process.
	This concern can be mitigated  using \emph{proofs}: artifacts
        that are checkable by an external and reliable proof checker,
        and which attest to the correctness of the verification process.
	However, such proofs tend to be extremely large, limiting their use in many scenarios.
	In this work, we address this problem by minimizing proofs of unsatisfiability produced by DNN verifiers.
	We present algorithms that remove facts which were learned
        during the verification process, but which are unnecessary for the proof itself.
	Conceptually, our method analyzes the dependencies among facts
        used to deduce \unsat{}, and removes facts that did not contribute. We then further minimize the proof by eliminating remaining unnecessary dependencies, using two alternative procedures.
	We implemented our algorithms on top of a proof producing DNN verifier, and evaluated them across several benchmarks.
	Our results show that our best-performing algorithm reduces proof size by 
        $37\%-82\%$ and proof checking time 
        by $30\%-88\%$, while introducing a
        runtime overhead of  $7\%-20\%$  to the verification process itself.
\end{abstract}

\section{Introduction}
\label{sec:Introduction}
Deep neural networks (DNNs) have made great strides in recent years,
becoming the state-of-the-art solution for a variety of tasks in
medicine~\cite{RaChBaTo22}, autonomous
driving~\cite{BoDeDwFiFlGoJaMoMuZhZhZhZi16}, natural language
processing~\cite{Op22}, and many other domains. However, DNNs lack
structure that humans can readily interpret, rendering them opaque and
potentially jeopardizing their trustworthiness~\cite{ChDiHiHaNuRiRuTr18}.  One prominent example
is the sensitivity of DNNs to small input perturbations, which can
lead to significant and undesirable changes to the networks'
outputs. This sensitivity can be exploited
maliciously~\cite{SzZaSuBrErGoFe13}, possibly leading to catastrophic
results.

The pervasiveness of DNNs, combined with their potential
vulnerability, has made DNN verification a growing research field
within the verification
community~\cite{BrMuBaJoLi23,BrBaJoWu24,AkKeLoPi19,AvBlChHeKoPr19,BaShShMeSa19,GeMiDrTsCHVe18,HuKwWaWu17,LyKoKoWoLiDa20,PuTa10,SaDuMo19,GaGePuVe19,TrBaXiJo20,WaPeWhYaJa18,ZhShGuGuLeNa20,GoPlSeRuSa21}.
Modern DNN verifiers typically use techniques such as SMT
solving~\cite{KaBaDiJuKo21,AbKe17,BaTi18,DeBj11}, abstract
interpretation~\cite{WaZhXuLiJaHsKo21,LyKoKoWoLiDa20,GaGePuVe19,GeMiDrTsCHVe18,GoPlSeRuSa21},
LP solving~\cite{TjXiTe17}, and combinations thereof.  Many such
verifiers are available~\cite{BrMuBaJoLi23,BrBaJoWu24}, and some tools
have even been applied to industrial case
studies~\cite{ElElIsDuGaPoBoCoKa24,AmFrKaMaRe23,KoLeEdChMaLo23}. Given
a DNN and a property over its input and outputs, DNN verifiers will
typically attempt to either find an input to the network that
violates the property or conclude that so such input exists.

Despite this success, one significant issue that the DNN verification
community faces is related to the \emph{reliability} of verifiers. One
concern is that even state-of-the-art verifiers may contain
implementation bugs. A second concern is that verifiers typically use
floating-point arithmetic, which is crucial for scalability but could lead to rounding errors and
numerical instability issues. These issues might be exploited to
compromise a verifier's soundness~\cite{JiRi21, ZoBaCsIsJe21}, thus
undermining user trust in these tools.
 
Due to their complexity, verifying the correctness of DNN verifiers
themselves is typically infeasible. Instead, DNN verifiers can
produce \emph{proofs} for their verification results,
as is commonly done in SAT~\cite{HeBi15} and SMT~\cite{BaDeFo15}.
These proofs, which serve as witnesses of correctness of the
verification result, constitute a mathematical object and can be
checked by an external, trusted~\emph{proof checker}~\cite{
  DeIsKoStPaKa25, DeIsPaStKoKa23}.  When a violation of the property
is discovered, the returned counter-example may serve as a proof of
correctness. However, proving the absence of a violation is less
straightforward, due to the NP-hardness of the
problem~\cite{KaBaDiJuKo21, SaLa21,IsZoBaKa23}.

So far, only a handful of
attempts have been made to enable DNN verifiers to produce
proofs~\cite{IsBaZhKa22}. These approaches typically include eager bookkeeping
of many intermediate lemmas deduced during
verification~\cite{IsBaZhKa22}. Therefore, the resulting proofs tend
to be extremely large, increasing the memory consumption of the
verifier and harming the performance of proof
checkers~\cite{DeIsKoStPaKa25}. These problems limit the applicability
of proof-producing verifiers, and call for further analysis to
minimize proof sizes.

In this work, we address this problem by minimizing the proof objects
constructed using the framework of Isac et al.~\cite{IsBaZhKa22}. Our proof reduction method involves two main steps. First, we apply a \emph{dependency analysis} procedure that recursively examines the dependencies of each lemma on previously learned lemmas. This helps to identify and remove lemmas that were eagerly learned and stored during the verification process, but are not actually required in the final proof.
Second, we apply two alternatives for \emph{dependency minimization} procedures, which
further refine the proof by removing unnecessary lemma dependencies,
keeping only a minimal subset of dependencies for each lemma. This
results in additional lemmas being classified as unused, leading to a
smaller proof. 

After the termination of the dependency analysis and minimization, the unused lemmas can be safely removed from the proof, without needing to be checked. Furthermore, our method allows lemma deletion on-the-fly, thus avoiding accumulation of unused lemmas and minimizing the memory consumption during the verification process. This, in turn, increases the scalability of proof producing DNN verifiers.

We evaluated our method across multiple benchmarks,
using~\cite{IsBaZhKa22} as a baseline. We measured the
effectiveness of our algorithms along two axes: the size of the
resulting proof, and the time
overhead for applying our algorithms. The results  for our best-performing algorithm indicate reduction
of proof size of 37\%--82\% on average and of proof
checking time by 30\%--88\% on average. The reduction incurs a
overhead of  7\%--20\%  to the verifier's runtime, where for some benchmarks the average verification time was slightly improved, possibly due to the use of more efficient data structures.

Our contributions are summarized as follows:
\begin{inparaenum}[(i)]
	\item an algorithm for analyzing the dependencies of each lemma learned within proof objects, removing unused lemmas;
	\item two algorithms for further detecting a minimal
          subset of dependencies; and
	\item a demonstration of a substantial reduction of the proof size over several benchmarks, while  adding a reasonable overhead in performance speed.
\end{inparaenum} 

The rest of this paper is organized as follows:
In~\Cref{sec:Background} we provide the necessary background on DNNs
and their verification, and on proof production. In~\Cref{sec:contrib} we explain our method to reduce the proof size. \Cref{sec:Evaluation} is dedicated to evaluation of our method over several benchmarks, with respect to both proof size and verification speed.  In~\Cref{sec:relwork} we discuss related work, and lastly, we conclude  and outline ideas for future research in~\Cref{sec:conclusion}.

\section{Background}
\label{sec:Background}

\subsection{Deep Neural Networks and their Verification}
\mysubsection{Deep Neural Networks (DNNs)~\cite{GoBeCu16}.} 
Formally, a DNN $ \nn:\rn{n}\rightarrow\rn{m} $ is a sequence
of $ k $ layers $ L_0,...,L_{k-1} $ where each layer $ L_i $ consists of
$ s_i \in \mathbb{N} $ nodes $ v^1_i,...,v^{s_i}_i $ and $s_i$ biases
$b^j_i\in\mathbb{Q}$. Each directed edge in the DNN is of the form
$(v^l_{i-1},v^j_i)$ and is labeled with a weight $w_{i,j,l}\in\mathbb{Q}$.
The assignment to the nodes in the input layer is defined by $v^j_0 = x_j$, where $\overline{x}\in\rn{n}$ is the input vector. The assignment for the
$ j^{th} $ node in the $ 1 \leq i < k-1 $ layer is computed as
$v^j_i = f_i \left( \underset{l=1} { \overset {s_{i-1}} { \sum } }
w_{i,j,l} \cdot v^l_{i-1} + b^j_i \right)$
for some activation function
$f_i:\mathbb{R}\rightarrow\mathbb{R}$.
Neurons in the output layer are computed similarly, where $f_{k-1}$ is the identity function.

One of the most common activation functions is the
\textit{rectified linear unit} (\relu), defined as $\relu(x) \coloneq
\max(x,0)$. 
The function has two linear phases: if the input is non-negative, the functions is \emph{active}; otherwise,
the function is \emph{inactive}. 
For simplicity, we focus on DNNs with
\relu{} activations, though our work can be extended to support any piecewise-linear activation (e.g., \emph{maxpool}).

\begin{Example} A simple DNN with four layers appears in~\Cref{fig:toyDnn}. For
simplicity, all biases are set to $0$ and are omitted. For input $\langle 2,1\rangle$, the node in the second layer
evaluates to
$ \relu(2 \cdot 1 \; + \; 1 \cdot (-1) ) = \relu(1) = 1 $; the nodes
in the third layer evaluate to $ \relu(1 \cdot (-2)) = 0 $ and $
\relu(1 \cdot 1) = 1 $; and the node in the output layer evaluates to $ 0 + 1 \cdot 2 = 2$.
\end{Example}
\begin{figure}
		\vspace{-1.0cm}
	\begin{center}
		\scalebox{0.82}{
			\def\layersep{2cm}
			\begin{tikzpicture}[shorten >=1pt,->,draw=black!50, node distance=\layersep,font=\footnotesize]
				
				\node[input neuron] (I-1) at (0,-1) {$x_1$};
				\node[input neuron] (I-2) at (0,-3) {$x_2$};
				
				\node[hidden neuron] (H-1) at (\layersep, -2) {$v_1$};
				
				\node[hidden neuron] (H-2) at (2*\layersep,-1) {$v_2$};
				\node[hidden neuron] (H-3) at (2*\layersep,-3) {$v_3$};
				
				\node[output neuron] (O-1) at (3*\layersep,-2) {$y$};
				
				\draw[nnedge] (I-1) --node[above,pos=0.4] {$1$} (H-1);
				\draw[nnedge] (I-2) --node[below,pos=0.4] {$-1$} (H-1);

				\draw[nnedge] (H-1) --node[above] {$-2$} (H-2);
				\draw[nnedge] (H-1) --node[below] {$1$} (H-3);
				\draw[nnedge] (H-2) --node[above] {$1$} (O-1);
				\draw[nnedge] (H-3) --node[below] {$2$} (O-1);
				
				\node[above=0.05cm of H-1] (b1) {$\relu{}$};
				\node[above=0.05cm of H-2] (b1) {$\relu{}$};
				\node[above=0.05cm of H-3] (b1) {$\relu{}$};
				
			\end{tikzpicture}
		}
	\end{center}
		\vspace{-0.4cm}
	\caption{A toy DNN with \relu{} activations.}
	\label{fig:toyDnn}
	\vspace{-0.6cm}
\end{figure}
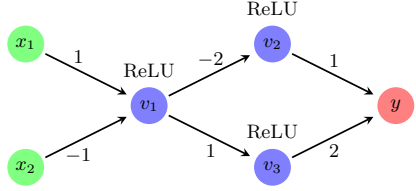

\mysubsection{The DNN Verification Problem.}
Consider a DNN $ \nn:\rn{n}\rightarrow\rn{m}$ where
$\overline{x}\in\rn{n}, \overline{y}\in\rn{m}$ denote the network's
inputs and outputs. The \emph{DNN verification problem} is the problem
of deciding whether there exist
$(\overline{x},\overline{y})\in\rn{n+m}$  such that
$ (\nn(\overline{x}) = \overline{y}) \wedge \varphi_I(\overline{x})\wedge
\varphi_O(\overline{y})$ for some
property $\varphi_I\wedge\varphi_O$.
If such
$\overline{x},\overline{y}$ exist, we say the problem is
satisfiable, denoted \sat{}. Otherwise, we say that it is unsatisfiable, denoted \unsat.
In this work, we consider \emph{quantifier-free linear properties}, i.e., formulas without quantifiers whose atoms are linear equations and inequalities.

\begin{Example}\label{ex:property} Consider the DNN
  from~\Cref{fig:toyDnn}, the input property
  
  \[
    \varphi_I(\overline{x}) \coloneq (1 \leq x_1) \wedge (1 \leq x_2)
    \wedge (x_1 \leq 2) \wedge (x_2 \leq 2)
  \]
  and the output property
  \[
    \varphi_O(y) \coloneq (y \leq -1).
  \]
  This instance is \unsat, as the output is a sum of two non-negative components with positive weights.
  	\vspace{-0.15cm}
\end{Example}

\mysubsection{Linear Programming (LP)~\cite{Da63} and DNN Verification.} In LP we seek an
assignment to a set of variables that satisfies a set of linear
constraints, while maximizing a given objective function.
As is common in the context of DNN verification, we assume here that the
objective function is constant and ignore it.
Formally,
let
$V = \left[x_1,\ldots,x_n\right]^\intercal\in \rn{n}$ denote 
variables, $ A \in M_{m \times n}( \mathbb{R})$ denote a constraint matrix
and $l, u \in (\mathbb{R}\cup \lbrace \pm \infty \rbrace)^{n} $
denote vectors of bounds. The LP problem is to decide the satisfiability of
$(A\cdot V = 0) \wedge (l\leq V \leq u)$.  Throughout the paper, we
use $l(x_i)$ and $u(x_i)$, to refer to the lower and upper bounds
(respectively) of $x_i\in V$, and denote $l \leq u$ to indicate that for each element $i$, $l(x_i) \leq u(x_i)$.

A DNN verification query can be encoded as a tuple
$Q=\langle V,A,l,u,R\rangle$, where $\langle V,A,l,u \rangle$
constitutes a linear program~\cite{Da63,KaBaDiJuKo21} that represents
the property and the affine transformations within $\nn$, and where
$R$ is the set of \relu{} activation constraints of the form
$f_i=\relu(b_i)$, for variables $b_i,f_i\in V$. In addition, a fresh
auxiliary variable $a_i$ is added for each \relu{} constraint,
representing the non-negative difference of its output and input. The
variable is added to $V$, with an equation $a_i = f_i -
b_i$ added to $A$, and the bounds $l(a_i)= 0$, $u(a_i) = u(f_i) - l(b_i)$ added to $l,u$ respectively.

\begin{Example} We
	demonstrate this encoding with the example presented
	in~\Cref{fig:toyDnn} and the property from~\Cref{ex:property}. Let
	$x_1, x_2$ denote the network's inputs and $y$ denote its output,
	and let $b_i,f_i$ ( $i\in\{1,2,3\}$) denote the input and
	output of neuron $v_i$, respectively.  The affine transformations
	of the network are reduced to the equations:
	\vspace{-0.5cm}
	\begin{align*}
		x_1 - x_2 -b_1 &= 0\\
		2f_1 + b_2 &= 0\\
		f_1 - b_3 &= 0\\
		f_2 + 2f_3 - y &= 0
	\end{align*}
	accompanied by $f_i - b_i - a_i = 0$ for $i\in\{1,2,3\}$. The property
	is encoded using the inequalities:
	$
	(1 \leq x_1), (1 \leq x_2),
	(x_1 \leq 2), (x_2 \leq 2), (y \leq -1).
	$
	We also add the
	inequalities $(0 \leq f_i)$ (for $i\in\{1,2,3\}$) to express the
	non-negativity of \relu{}s, and $(0 \leq a_i)$ as described. The aforementioned constraints form the LP
	part of the query; and its piecewise-linear portion is comprised of
	the constraints
	$f_i = \relu(b_i)$ for $i\in\{1,2,3\}$.
\end{Example}

Using this
encoding, the verification query can be dispatched through a series of
invocations of an LP solver, combined with \emph{case splitting} to
handle the piecewise-linear constraints~\cite{KaBaDiJuKo21}.
Schematically, the DNN verifier begins by invoking the LP solver over
the linear part of the query. An \unsat{} result of the LP solver
implies the overall query is \unsat. Otherwise, the LP solver finds an
assignment that satisfies the linear part of the query. If this assignment
satisfies the piecewise-linear part of the query, the DNN verifier may
conclude \sat. If neither case applies, the DNN verifier splits the
query into two subqueries, by deciding the phase of a single \relu{}
constraint $f_i=\relu{}(b_i)$. One subquery is augmented with the bounds $b_i \geq 0$ and $a_i \leq 0$
corresponding to the active phase; the other is augmented with $b_i \leq 0$ and $f_i \leq0$, corresponding to the inactive
phase. Then, the LP solver is invoked again over each subquery, and
the process is repeated. Note that the introduction of auxiliary
variables is intended to reduce case splitting to bound updates,
without needing to add new equations; empirically, this is known to
improve performance~\cite{WuIsZeTaDaKoReAmJuBaHuLaWuZhKoKaBa24}.

This splitting approach creates an abstract tree structure, where each node represents a case split. 
If the LP solver deduces \unsat{} for each leaf of the tree, then the original query is \unsat{} as well. If not, then due to the completeness of LP solvers, it finds a satisfying assignment for one of the subqueries.

The search tree may be exponentially large in the number of 
piecewise linear constraints (i.e., the number of neurons). Therefore, it is common to use case splitting rarely, and apply algorithms to deduce tighter bounds of
variables. These algorithms may conclude in advance that some neurons'
phases are fixed, i.e., they are always active or  inactive,
avoiding the need to split on them~\cite{KaBaDiJuKo21}.


\subsection{Proof Production for DNN Verification}
When DNN verification is regarded as a satisfiability problem, a
satisfying assignment is an efficient proof of \sat: one can simply
run it through the network and observe that it satisfies the given property. The \unsat{} case is more complex, and requires producing
a \emph{proof certificate}. The first method
for proof production in DNN verification performed via LP solving and
case splitting was introduced in~\cite{IsBaZhKa22}. These proof
certificates consist
of a \emph{proof tree} that replicates the search tree created by the
verifier, and is constructed during the verification process. Recall
that for an \unsat{} query, all leaves of the search tree represent
\unsat{} subqueries. Thus, each internal node of the proof tree
represents a case split performed by the verifier, and each leaf
corresponds to a subquery for which the verifier has deduced \unsat. Therefore a \emph{proof checker} for these proofs is required to traverse the tree, check its structural correctness (i.e., all splits are correct), and certify the unsatisfiability of each leaf.

In each leaf of the search tree, the verifier was able to conclude
\unsat{} using solely the available linear constraints.
Thus, a proof for the leaf's unsatisfiability can be described as a vector, according
to the Farkas lemma~\cite{ChvatalLP}, which combines the linear
constraints to derive an evident, easily-checkable, contradiction~\cite{IsBaZhKa22}.
More formally:

\begin{theorem}[Farkas Lemma Variant (From~\cite{IsBaZhKa22})]
	\label{thm:Farkas}
	Observe the constraints $ A\cdot V = \bar{0} $ and $ l \leq V \leq u $,	where $ A \in M_{m \times n} (\mathbb{R})$  and $ l,V,u \in \rn{n} $. The constraints are \unsat{} if and only if $ \exists w \in \rn{m} $ such that for 
	$w^\intercal \cdot A \cdot V \coloneq \underset{i=1}{\overset{n}{\sum}}c_i\cdot x_i$, we have that  $ \underset{c_i > 0}{\sum}c_i\cdot u(x_i) +  \underset{c_i < 0}{\sum}c_i\cdot l(x_i) < 0$
	whereas $  w^\intercal \cdot \bar{0} = 0 $. Thus,
	$w$ is a proof of the constraints' unsatisfiability, and can be constructed during LP solving.
\end{theorem}

\begin{Example}\label{ex:query} We demonstrate this theorem with our
example. Consider the query based on the DNN in~\Cref{fig:toyDnn} and
the property in~\Cref{ex:property}. Assume, for simplicity, that all
lower and upper bounds not explicitly stated are $-1$ 
or $1$, respectively. Further assume that both $v_2,v_3$ are inactive (i.e. $f_2=f_3=0$ and $b_2,b_3\leq 0$).
The linear part of the query yields the LP instance:
\[
\scalebox{0.92}{$
A =
\begin{bmatrix}
	1 & -1 & -1 & 0  & 0  & 0  & 0  & 0 & 0 & 0 & 0 & 0\\
	0 &  0 &  0 & 1 & 0 & 2 &  0 & 0 & 0 & 0 & 0 & 0 \\
	0 &  0 &  0 &  0 & -1  & 1  & 0& 0 & 0 & 0 & 0 & 0\\
	0 &  0 & 0 & 0  & 0 & 0  &  1 & 2 & 0 & 0 & 0 &-1 \\
	0 &  0 &  -1 &  0 & 0  & 1  & 0& 0 & -1 & 0 & 0 & 0\\
	0 &  0 &  0 &  -1 & 0  & 0  & 1& 0 & 0 & -1 & 0 & 0\\
	0 &  0 &  0 &  0 & -1  & 0 & 0& 1 & 0 & 0 & -1 & 0\\
\end{bmatrix} \quad
\begin{aligned}
	u =&
	\begin{bmatrix}
		2 & \;\;2 & \;\;\;1  & \;0 & \;\;\;0  &\; 1 & \;\;0 &\; 0 &\;\; 2 & \;\;1 & 1& -1
	\end{bmatrix}^\intercal\\
	V = &
	\begin{bmatrix}
		x_1 & x_2 & \; b_1 & b_2 & \; b_3 & f_1 & f_2 & f_3 & a_1 & a_2 & a_3 & y
	\end{bmatrix}^\intercal \\
	l =&
	\begin{bmatrix}
		1 & \;\;1 & \, -1  & -1  & -1  & 0 & \;0 & \;0 &\;\; 0 &\;\; 0 & 0 & -1
	\end{bmatrix}^\intercal 
\end{aligned}
$}
\]
Note that in $A$ the first four rows corresponds to the affine transformations of the DNN, and the latter three rows are introduced with the auxiliary variables.
Consider the vector
\[
  w = 	\begin{bmatrix}
	0 & 0 & -1  & -2 & 0 & 0 & 0
  \end{bmatrix}^\intercal.
\]
The product $w^\intercal\cdot A\cdot V$ yields the equation
$
  -f_1+b_3-2f_2-4f_3+2y = 0.
$
The left hand side of the equation is at most:
$
  -l(f_1)+u(b_3)-2\cdot l(f_2)-4\cdot l(f_3)+2\cdot u(y) = -2 < 0.
$
 Therefore, according to~\Cref{thm:Farkas}, $w$ is a proof of \unsat. Note that there may be multiple proofs of this result.
An overall proof of \unsat{} for the query consists of a proof tree, where each leaf has a proof vector proving \unsat{} for the corresponding subquery. 
\end{Example}

\mysubsection{Bound Tightening Lemmas.}  DNN verifiers often employ bound
tightening procedures that deduce bounds using the non-linear
activations, called \emph{bound tightening lemmas}.
Each
bound tightening lemma consists of a ground bound and a learned bound,
which replaces a bound currently in $l$ or $u$.
As opposed to deduction of bounds using linear equations, proving such lemmas generally cannot be captured by the Farkas
vector described above, and require separate proofs~\cite{IsBaZhKa22}. For example, given $f = \relu(b)$
and present bounds $f\leq 5$ and $b\leq 7$, it is possible to deduce
that $b\leq 5$. The lemma can be proven using a Farkas vector
that proves $f\leq 5$, and another proof rule that
captures this form of bound derivation.

\begin{Example}\label{ex:lemma} Consider the query $Q$ in~\Cref{ex:query}. We
can use the equation $b_2 = -2f_1$, which is equivalent to
$2f_1+b_2=0$, to deduce that $u(b_2) = -2\cdot l(f_1) = 0$. Then,
based on the \relu{} constraint, we deduce that the neuron $v_2$ is
inactive, i.e., that $u(f_2) = 0$. A lemma with a ground bound
$u(b_2) = 0$, learned bound $u(f_2) = 0$ and a proof vector for
generating $b_2 = -2f_1$ (the vector
$\begin{bmatrix} 0 & -1 & 0 & 0 & 0 & 0 & 0
\end{bmatrix}^\intercal$) is learned at the root of the proof tree, without any case splits. 
Similarly, we can deduce that $v_3$ is inactive. Observe that $f_3 = \frac{1}{2}y - \frac{1}{2}f_2$, and thus $u(f_3)= \frac{1}{2}u(y) - \frac{1}{2}l(f_2) = -0.5$.
Then a  lemma
with a ground bound $u(f_3) = -0.5$, learned bound $u(b_3) = 0$ and a proof vector $\begin{bmatrix} 0 & 0 & 0 & -0.5 & 0 & 0 & 0
\end{bmatrix}^\intercal$ is learned. Combined with~\Cref{ex:query}, we get a proof of \unsat{} for $Q$ that consists of a single node with two lemmas, as illustrated in~\Cref{fig:proof}.
\end{Example}
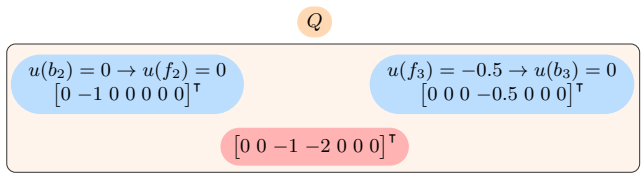
\begin{figure}[h!]
	\vspace{-0.8cm}
	\centering
	\scalebox{0.83} {
		\def\ySep{1cm}
		\def\xSep{3cm}
		\begin{tikzpicture}[ >=stealth,shorten >=1pt,shorten <=1pt]
			\node[proofNode,fill =	orange!30] (root) at (0,0) []  {$Q$};
			\node[lemmaNode, align=center] (lemma1) at (-\xSep, -\ySep)  [] {$u(b_2) = 0 \rightarrow u(f_2) = 0$ \\ $\begin{bmatrix} 0 & -1 & 0 & 0 & 0 & 0 & 0 \end{bmatrix}^\intercal$};
			\node[lemmaNode,align=center] (lemma2) at (\xSep, -\ySep)  [] {$u(f_3) = -0.5 \rightarrow u(b_3) = 0$ \\ $\begin{bmatrix} 0 & 0 & 0 & -0.5 & 0 & 0 & 0 \end{bmatrix}^\intercal$};
			\node[proofNode, align=center] (unsat) at (0, -2*\ySep)  [] {$\begin{bmatrix} 0 & 0 & -1 & -2 & 0 & 0 & 0 \end{bmatrix}^\intercal$};
			
			\begin{pgfonlayer}{background}
				
				\draw[rounded corners, fill =
				orange!10, opacity=0.8]
				($(root.south west)
				+ (-4.9cm, -0.1cm ) $)
				rectangle 
				($(unsat.south east) + ( 4cm, -0.1cm) $);
				
			\end{pgfonlayer}
			
		\end{tikzpicture}
	}
	\caption{A proof tree example with a single node (orange). Given the query $Q$, two lemmas (blue) are learned before deriving \unsat{} (red).}
	\label{fig:proof}
	\vspace{-0.6cm}
\end{figure}

To avoid ambiguity in the definition of the bound vectors $l,u$, unless stated otherwise, $l,u$ denote the tightmost bound vectors, while $l_{input}, u_{input}$ denote the original bound vectors, given as inputs.
                             
\section{Proof Dependency Analysis and Minimization}
\label{sec:contrib}
In this section, we elaborate on our algorithms for proof
minimization. Recall that each lemma and each leaf within the proof
includes a proof vector; and that all these vectors have the same dimension, since the number of
rows in $A$ is constant. As proving \unsat{} of all leaves in the
proof tree is necessary for proving \unsat{} for the whole query,
proof minimization is achieved by reducing the number of lemmas used
within the proof. The minimization algorithms are applied upon deducing \unsat{} on each proof tree leaf, thus enabling on-the-fly minimization and preventing the accumulation of unnecessary lemmas.
 We begin by describing our algorithm for dependency
analysis in~\Cref{sec:proofanalysis} and two algorithms for dependency minimization
in~\Cref{sec:proofmin}. We conclude in~\Cref{sec:misc}, by introducing an additional minimization technique, which can be used in both algorithms, by explaining  the procedure
for removing the unused lemmas and by introducing engineering optimizations to improve performance speed.

\subsection{Analyzing Proof Dependencies}
\label{sec:proofanalysis}
We now describe our algorithm for minimizing proof-trees by dependency analysis. Given a proof vector --- either a proof of unsatisfiability for a search-tree leaf, or a
proof for a bound tightening lemma --- our analysis identifies which bound tightening lemmas are involved in its proof checking process.
Our approach for analyzing the dependencies of a given proof vector
$w$ appears as~\Cref{alg:lemmadeduction}.  

Conceptually,
\Cref{alg:lemmadeduction} reconstructs the linear combination
$w^\intercal\cdot A \cdot V$, denoted as
$\underset{i=1}{\overset{n}{\sum}}c_i\cdot x_i$. Then, as the proof
only uses the bounds $l(x_i)$ for variables with a negative
coefficients, and $u(x_i)$ for variables with a positive coefficient,
their corresponding lemmas form the \emph{dependency list} for
$w$. The dependency list provides a sufficient set of lemmas for
proving the unsatisfiability of $w$ --- and this set can potentially
be minimized.  Then, the algorithm further analyses the dependencies
of the discovered lemmas, only after we know they are used in the
proof. Recall that each lemma has its own proof vector, and thus this
analysis is carried out through a recursive call on each lemma's proof.

To expedite the algorithm, we store for every bound
$l(x_i),u(x_i)$, the lemma or split that was used to deduce it. In addition, to avoid cycles and ensure the completeness of our algorithms, each
lemma is given a unique, chronological ID. Whenever a lemma with ID $k$ is analyzed, only lemmas with
ID $l < k$ are considered in the analysis. 

\begin{algorithm}[t!]
	\textbf{Input:} A DNN verification subquery $Q=\langle V,A,l,u,R\rangle$, a list of learned lemmas $Lem$, and a proof of \unsat{} $w$ \\
	\textbf{Output:} A list of lemmas sufficient to deduce \unsat 
	\caption{\emph{proofDeps():} Construct the dependency list of a proof vector}
	\label{alg:lemmadeduction}
	\begin{algorithmic}[1]
		\Statex{// Let $\underset{i=1}{\overset{n}{\sum}}c_i\cdot
			x_i$ denote the linear combination $ w^\intercal \cdot A \cdot V$}
		\State {$Deps \leftarrow \emptyset$} \Comment{An empty list of lemmas}
		\For{$i\in [n] $}
		\If {$c_i > 0$ and $u(x_i)$ is learned by  $\ell^u_i\in Lem$}
		\State{$Deps \leftarrow Deps \cup \ell^u_i$}
		\ElsIf {$c_i < 0$ and $l(x_i)$ is learned by  $\ell^l_i\in Lem$}
		\State{$Deps \leftarrow Deps \cup \ell^l_i$}
		\EndIf
		\EndFor
		\For{$lem \in Deps$}
		\Comment{Repeat recursively}
		\State{$lem.includeInProof \leftarrow true$ }
		\State{$l',u' \leftarrow$ bounds with ID at most $lem.getID()$}
		\State{$w' \leftarrow lem.getProof()$ }
		\State{$Deps' \leftarrow proofDeps(\langle V,A',l',u',R\rangle, Lem, w')$}
		\State{$Deps \leftarrow Deps \cup Deps'$}
		\EndFor\\
		\Return{$Deps$}
	\end{algorithmic}
\end{algorithm}
\begin{Example}\label{ex:deplist} Recall that in~\Cref{ex:query}
  and~\Cref{ex:lemma}, after learning two lemmas we are able to derive \unsat{} from a proof vector $w = 	\begin{bmatrix}
	0 & 0 & -1  & -2 & 0 &0
\end{bmatrix}^\intercal$. Using this vector generates the equation $-f_1+b_3-2f_2-4f_3+2y = 0$, which uses the bounds $l(f_1),u(b_3),l(f_2),l(f_3),u(y)$ to prove \unsat.
As $u(b_3)$ is deduced using one of the lemmas in~\Cref{ex:lemma},
this lemma forms the dependency list of $w$. Then, the algorithm
recursively generates the dependency list of the proof vector of the
lemma in a similar manner. As shown in~\Cref{ex:lemma}, that lemma is
not dependent on $u(f_2)$. Thus, the lemma for deducing $u(f_2)$ is
not actually used in the proof of \unsat{}, and can then be removed from the proof
tree. An illustration of this example appears in~\Cref{fig:proofanalysis}. 
\end{Example}
\begin{figure}[h!]
	\vspace{-0.8cm}
	\centering
	\scalebox{0.83} {
		\def\ySep{1cm}
		\def\xSep{3cm}
		\begin{tikzpicture}[ >=stealth,shorten >=1pt,shorten <=1pt]
			\node[proofNode,fill =	orange!30] (root) at (0,0) []  {$Q$};
			\node[lemmaNodeRemoved, align=center] (lemma1) at (-\xSep, -\ySep)  [] {$u(b_2) = 0 \rightarrow u(f_2) = 0$ \\ $\begin{bmatrix} 0 & -1 & 0 & 0 & 0 & 0 & 0 \end{bmatrix}^\intercal$};
			\node[lemmaNode,align=center] (lemma2) at (\xSep, -\ySep)  [] {$u(f_3) = -0.5 \rightarrow u(b_3) = 0$ \\ $\begin{bmatrix} 0 & 0 & 0 & -0.5 & 0 & 0 & 0 \end{bmatrix}^\intercal$};
			\node[proofNode, align=center] (unsat) at (0, -2*\ySep)  [] {$\begin{bmatrix} 0 & 0 & -1 & -2 & 0 & 0 & 0 \end{bmatrix}^\intercal$};
			
			\draw[dashed,->, line width=0.5mm] (root) -- (lemma1) ;
			\draw[dashed,->, line width=0.5mm] (root) -- (lemma2) ;
			\draw[dashed,->, line width=0.5mm] (root) -- (unsat) ;
			
			\draw[dashed,->, line width=0.5mm] (\xSep, -1.45*\ySep) -- (0, -1.7*\ySep) ;

			\begin{pgfonlayer}{background}
				
				\draw[rounded corners, fill =
				orange!10, opacity=0.8]
				($(root.south west)
				+ (-4.9cm, -0.1cm ) $)
				rectangle 
				($(unsat.south east) + ( 4cm, -0.1cm) $);
				
			\end{pgfonlayer}
			
		\end{tikzpicture}
	}
	\caption{Using~\Cref{alg:lemmadeduction} over the proof vector used to derive \unsat{} (red) shows that a single lemma (blue) is used during the proof process, and the other (gray) is not. All proof vectors use information in the root query $Q$ (orange)}
	\label{fig:proofanalysis}
	\vspace{-0.6cm}
\end{figure}
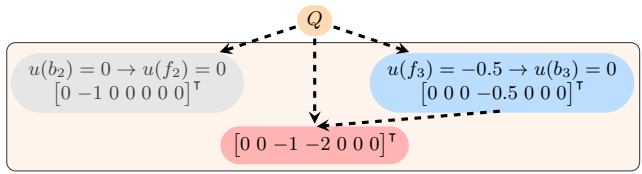

\subsection{Minimizing the Number of Proof Dependencies}
\label{sec:proofmin}
We now elaborate on our two alternatives of algorithms for minimizing the dependency list
of a proof vector. By doing so, we aim to increase the number of lemmas that
 do not appear in any dependency list, i.e., are unused. Such lemmas
 can then  be omitted from the proof, further reducing its overall size. 

 Let $w$ be a proof vector proving the unsatisfiability of a search-tree
 leaf. Our algorithm focuses on minimizing the number of bound
 tightening lemmas that participate in the refutation. Recall that a
 refutation is checked by computing the upper bound of
 $ w^\intercal \cdot A \cdot V =
 \underset{i=1}{\overset{n}{\sum}}c_i\cdot x_i$, i.e.,
\[\Delta \coloneq \underset{c_i > 0}{\sum}c_i\cdot u(x_i) +  \underset{c_i < 	0}{\sum}c_i\cdot l(x_i)
\]
 and ensuring it is negative. The idea of our algorithm lies in the fact that $\Delta$ can be very far from zero. In this case, we can relax some of the bounds derived from the participating bound lemmas, without altering the sign of $\Delta$, and thus maintain the proof's correctness.
 Analyzing the dependencies of a proof vector $\hat{w}$, used to prove a bound tightening lemmas, is performed similarly. In this case, $\Delta$ represents the difference between the bound it is used to prove and the actual bound achieved with $ \hat{w}^\intercal \cdot A \cdot V$.
 
\subsubsection{Proof Minimization.}
The first algorithm for proof minimization is greedy, and removes the maximal number of unnecessary dependencies for each proof vector separately. 
To do so, we first define for each lemma its \emph{contribution to $w$}, which can be seen as its effect on $\Delta$. More formally, given the input bound vectors $l_{input},u_{input}$, and the bound vectors $l,u$ with bounds derived by a set of bound lemmas and case splits. For each $c_j > 0$, the contribution of $u(x_j)$ to $w^\intercal \cdot A \cdot V$ is defined as: 
\[cont(u_j,w)\coloneq c_j\cdot(u(x_j) - u_{input}(x_j))\] For each  $c_j < 0$, the contribution of $l(x_j)$ to $w^\intercal \cdot A \cdot V$ is: \[cont(l_j,w)\coloneq c_j\cdot(l_{input}(x_j) - l(x_j))\]
Note that in all cases the contribution is non-positive, as $l_{input} \leq l \leq u \leq u_{input}$.
	
 These definitions naturally give rise to~\Cref{alg:proofmin}, which
 sorts the various lemmas according to their contributions to the
 proof, from small to large. Those lemmas with only a small
 contribution could potentially be omitted from the proof, without
 altering the sign of $\Delta$. Attempting to remove lemmas in this
 order ensures the removal of the maximum number of lemmas from the dependency list, and
 consequently the minimality of the output, as we prove
 in~\Cref{thm:minimality}. \Cref{alg:proofmin} can be used
 within~\Cref{alg:lemmadeduction}, before any recursive application of
 the algorithm. By that, \Cref{alg:proofmin} reduces the time overhead
 of~\Cref{alg:lemmadeduction}, by reducing the number of the latter's
 recursive calls.

\begin{theorem}
	\label{thm:minimality}
	Given $\langle V,A,l,u,R\rangle$, a proof vector $w$, a list of its lemma dependencies $D$, and the input bound vectors $l_{input},u_{input}$, \Cref{alg:proofmin} returns a dependency subset of $Deps$ that is minimal in size.
\end{theorem}

\begin{proof}
	Assume towards contradiction that there exists $D'$, which constitutes a dependency list of $w$, with $|D'| < |D|$. First, observe that if $D' \subsetneq D$, then there is some lemma $\ell \in D \setminus D'$ that can be removed from $D$ when considered in line $6$, this contradicts the definition of the algorithm.
	 
	Therefore, $D \setminus D'$ and $D' \setminus D$ are not empty. Let $\ell \in D \setminus D'$, with the minimal contribution $\delta$, and let $\ell' \in D' \setminus D$ with the minimal contribution $\delta'$.
	Since~\Cref{alg:proofmin} is greedy, we must have that $\delta' \leq \delta$, as otherwise the algorithm would have removed $\ell$ from $D$ before removing $\ell'$. This means that we can replace $\ell'$ with $\ell$ in $D'$ without increasing the overall contribution removed from $Deps$. We can repeat the process until saturation, i.e., when $D' \subsetneq D$, which leads to a contradiction. $\square$
\end{proof}

\begin{algorithm}[t!]
	\textbf{Input:} A DNN verification subquery  $Q=\langle V,A,l,u,R\rangle$, a proof vector $w$, a list of its lemma dependencies $Deps$, and the input bound vectors $l_{input},u_{input}$\\
	\textbf{Output:} A minimial dependency list of $w$
	\caption{\emph{proofMin()}: Minimize the dependency list of a proof vector}
	\label{alg:proofmin}
	\begin{algorithmic}[1]
		\Statex{// Let $\underset{i=1}{\overset{n}{\sum}}c_i\cdot
			x_i$ denote the linear combination $ w^\intercal \cdot A \cdot V$}
		\State{$\Delta\leftarrow\underset{c_i > 0}{\sum}c_i\cdot u(x_i) +  \underset{c_i < 0}{\sum}c_i\cdot l(x_i)$}
		\State{$Deps \leftarrow Deps.sortByContribution()$}
		\State{$\delta \leftarrow 0$}
		\For{$dep \in Deps$}
		\State{$\delta += |cont(dep.getBoundType(),w)|$} \Comment {Use $A,V,l,u,l_{input},u_{input}$}
		\If{$\delta \leq |\Delta|$}
		\State {$Deps.remove(dep)$}
		\Else{}
		\Return{$Deps$}
		\EndIf
		\EndFor
		\State\Return $\emptyset$
	\end{algorithmic}
\end{algorithm}

\begin{Example}
	In~\Cref{ex:deplist} we saw that the proof vector $w = 	\begin{bmatrix}
		0 & 0 & -1  & -2 & 0 &0
	\end{bmatrix}^\intercal$, presented in~\Cref{ex:query}, is dependent on only one of the two lemmas learned as in~\Cref{ex:lemma}. In addition, the analysis in~\Cref{ex:query} shows that $\Delta = -2$, and that the proof vector $w$ is dependent on the lemma used to deduce $u(b_3)=0$, whose input bound is $u(b_3)=1$. Therefore, the contribution of the lemma is $1\cdot(0-1)=-1$, indicating this lemma can be removed from the dependency list as well.
	Observe that indeed, given the input bounds, we have: 
	\[
	-l(f_1)+u_{input}(b_3)-2\cdot l(f_2)-4\cdot l(f_3)+2\cdot u(y) = -1 < 0
	\] Together with~\Cref{ex:deplist}, we conclude that both lemmas from~\Cref{ex:lemma} can be removed from the proof tree. An illustration for this example appears in~\Cref{fig:proofminanalysis}.
\end{Example}
\begin{figure}[h!]
	\centering
	\scalebox{0.83} {
		\def\ySep{1cm}
		\def\xSep{3cm}
		\begin{tikzpicture}[ >=stealth,shorten >=1pt,shorten <=1pt]
			\node[proofNode,fill =	orange!30] (root) at (0,0) []  {$Q$};
			\node[lemmaNodeRemoved, align=center] (lemma1) at (-\xSep, -\ySep)  [] {$u(b_2) = 0 \rightarrow u(f_2) = 0$ \\ $\begin{bmatrix} 0 & -1 & 0 & 0 & 0 & 0 & 0 \end{bmatrix}^\intercal$};
			\node[lemmaNodeRemoved,align=center] (lemma2) at (\xSep, -\ySep)  [] {$u(f_3) = -0.5 \rightarrow u(b_3) = 0$ \\ $\begin{bmatrix} 0 & 0 & 0 & -0.5 & 0 & 0 & 0 \end{bmatrix}^\intercal$};
			\node[proofNode, align=center] (unsat) at (0, -2*\ySep)  [] {$\begin{bmatrix} 0 & 0 & -1 & -2 & 0 & 0 & 0 \end{bmatrix}^\intercal$};

			\draw[dashed,->, line width=0.5mm] (root) -- (lemma1) ;
			\draw[dashed,->, line width=0.5mm] (root) -- (lemma2) ;
			\draw[dashed,->, line width=0.5mm] (root) -- (unsat) ;
			
			\draw[dashed,->, line width=0.5mm,,color=black!20] (\xSep, -1.45*\ySep) -- (0, -1.7*\ySep) ;
			
			\begin{pgfonlayer}{background}
				
				\draw[rounded corners, fill =
				orange!10, opacity=0.8]
				($(root.south west)
				+ (-4.9cm, -0.1cm ) $)
				rectangle 
				($(unsat.south east) + ( 4cm, -0.1cm) $);
				
			\end{pgfonlayer}
			
		\end{tikzpicture}
	}
	\caption{Using~\Cref{alg:proofmin} over the proof vector used to derive \unsat{} (red) shows that both lemmas (gray) are unnecessary for proving \unsat.}
	\label{fig:proofminanalysis}
	\vspace{-0.6cm}
\end{figure}
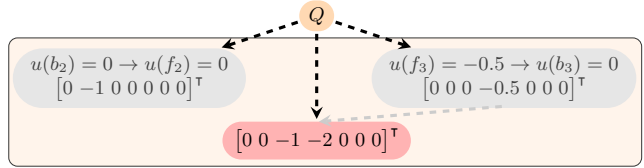

\subsubsection{Global Heuristic for Proof Minimization.}
\label{sec:globmin}
Although~\Cref{alg:proofmin} guarantees achieving the minimum number of dependencies for each proof vector, this does not necessarily ensure achieving the global minimum across the entire proof tree. To incorporate global considerations, we adjust the order in which lemmas are considered for removal, prioritizing those with more dependencies and whose elimination is likely to trigger the removal of additional downstream lemmas.  

Specifically, before applying the minimization algorithm, we recursively apply the dependency analysis procedure to each lemma and compute its total number of dependencies, including those discovered through the recursive applications of the minimization procedure. The lemmas are then sorted in decreasing order by their overall dependency count (i.e., lemmas with more dependencies are considered for removal first). In case of ties, we fall back on the previous contribution score.
We apply the analysis until saturation, meaning no further reduction of dependencies is possible. This ensures minimality of the subset of dependencies, though it does not guarantee the resulting subset is smallest. This can occur if \Cref{alg:globproofmin} diverges significantly from \Cref{alg:proofmin}, for example, when lemmas with high dependency counts contribute heavily. In such cases, the former algorithm may remove these lemmas first, while the latter will instead remove many more lemmas with a lower dependency count. 
To maintain the applicability of the algorithm during proof generation, we apply it each time an \unsat{} leaf is detected, considering all lemmas learned globally up to that point.

This intuition is formalized in~\Cref{alg:globproofmin}, and it opens up possibility to explore additional heuristics for prioritizing lemmas for removal, e.g., by considering different combinations of each lemma's number of dependencies and its contribution.

 Note that introducing global considerations for proof minimization involves dependency analysis for many lemmas, and thus might incur a relatively large time overhead.
 
 \begin{algorithm}[h!]
 	\textbf{Input:}  A DNN verification subquery  $Q=\langle V,A,l,u,R\rangle$, a proof vector $w$, a list of its lemma dependencies $Deps$, the input bound vectors $l_{input},u_{input}$, and the overall list of learned lemmas $Lem$\\
 	\textbf{Output:} A minimial dependency list of $w$
 	\caption{\emph{globProofMin()}: Minimize the dependency list of a proof vector}
 	\label{alg:globproofmin}
 	\begin{algorithmic}[1]
 		\Statex{// Let $\underset{i=1}{\overset{n}{\sum}}c_i\cdot
 			x_i$ denote the linear combination $ w^\intercal \cdot A \cdot V$}
 		\State{$\Delta\leftarrow\underset{c_i > 0}{\sum}c_i\cdot u(x_i) +  \underset{c_i < 0}{\sum}c_i\cdot l(x_i)$}
 		\State{$Deps \leftarrow Deps.sortByDepsLength(Q,Deps,Lem)$}
 		\State{$\delta \leftarrow 0$}
 		\State{$saturation \leftarrow false$}
 		\While {$!saturation$}
 		 \State{$saturation \leftarrow true$}
 		\For{$dep \in Deps$}
 		\State{$\delta += |cont(dep.getBoundType(),w)|$} \Comment{Use $A,V,l,u,l_{input},u_{input}$}
 		\If{$\delta \leq |\Delta|$}
 		\State {$Deps.remove(dep)$}
 		 \State{$saturation \leftarrow false$}
 		\Else\State{$\delta -= |cont(dep.getBoundType(),w)|$} \Comment{Revert}
 		\EndIf
 		\EndFor
 		\EndWhile
 		\State\Return $Deps$
 	\end{algorithmic}

	\vspace{0.2cm}
 	\begin{algorithmic}[1]
 			\Statex sortByDepsLength($Q,Deps,Lem$):
 			\For {$\ell\in Deps$}
 			\If {$!\ell.wasAnalyzed$}
 			\State $\ell.score \leftarrow proofDeps(Q,Lem,\ell.getProof()).size()$
 			\State $\ell.tieBreaker \leftarrow \ell.getContribution()$
 			\State $\ell.wasAnalyzed \leftarrow true$
 			\EndIf
 			\EndFor
				\State\Return Deps.sortByScore()
 	\end{algorithmic}
 \end{algorithm}
 
\subsection{Additional Improvements and Lemma Deletion}
\label{sec:misc}
 \mysubsection{Further Reduction of Proof Size.}
 	Both minimization algorithms presented above consider the
        dependencies for each proof vector separately. However, it is
        possible that minimizing dependencies at the proof-vector
        level leads to redundant dependencies at the proof-tree
        level. Consider the case where two proof vectors $w_1$, $w_2$
        depend on a lemma $\ell$, and that the dependencies of $w_1$ are
        minimized before those of $w_2$. If the dependency
        minimization of $w_1$ indicates that $\ell$ is retained in the
        proof, then removing it from the dependencies of $w_2$ will
        not lead to the removal of $\ell$ from the overall proof, and
        will unnecessarily increase the contributions counter $\delta$
        \Cref{alg:proofmin} (line 5).
 	Therefore, a straightforward improvement to both minimization algorithms is to consider removing dependencies only from lemmas that are not currently used in the proof. This approach allows us to skip lemmas that are already essential and would be retained regardless.
\begin{Example}
	 Consider a lemma $\ell_1$ with $\Delta_1 = -1$. Suppose that
         proof analysis of~\Cref{alg:lemmadeduction} determines that $\ell_1$ is dependent on $\ell_2,
         \ell_3$ with respective contributions of $-0.5$ and
         $-0.6$. Suppose also that other lemmas depend on $\ell_2$, so it is retained in the proof, but
         none depend on $l_3$ yet. In this case,~\Cref{alg:proofmin}
         would first remove the dependency of $\ell_1$ in $\ell_2$ and then
         stop as $|-0.5| + |-0.6| > |-1|$. Therefore, this application of~\Cref{alg:proofmin} will keep
         both $\ell_2$ and $\ell_3$ in the proof. Alternatively, by ignoring
         $\ell_2$ as it is already used in the proof, and removing the
        dependency on $\ell_3$ instead, we may reduce the overall number
        of lemmas required for the proof.
\end{Example}
\mysubsection{Avoiding Duplicate Computations.} Recall that~\Cref{alg:lemmadeduction} recursively iterates through lemmas and analyzes their dependencies. Since multiple lemmas can depend on the same lemma, it is useful to avoid redundant or repeated analysis. To address this, we add a flag to each lemma indicating whether it has already been analyzed. This ensures that each lemma is analyzed at most once.

\mysubsection{Unused Lemma Deletion.} To reduce the memory consumption
of the DNN verifier during the verification process, we propose to remove unused lemmas as early as possible. Since each lemma is applicable only in the search state $S$ where it was deduced and in all descendant states of $S$, it becomes irrelevant once the sub-tree rooted at $S$ has been fully explored. At that point, any lemma in $S$ that was not involved in any dependency analysis can be safely deleted.

\mysubsection{Supporting Additional Constraints.} Although so far we
focused on DNNs with \relu{} constraints, our minimization algorithms
can be applied to additional kinds of constraints in a
straightforward manner. Therefore, assuming the underlying verifier
supports proof production for additional kinds of piecewise-linear constraints, the
minimization algorithms are readily applicable. 

\section{Implementation and Evaluation}
\label{sec:Evaluation}

To assess the effectiveness of our proposed proof minimization techniques, we conducted an empirical evaluation across six benchmark suites in neural network verification (ordered alphabetically): \begin{inparaenum}[(i)]
	\item the \textsc{Acas-Xu} drone collision avoidance networks~\cite{JuKoOw19};
	\item the \textsc{cersyve} benchmark~\cite{YaHuWeLiLi25} for verification of neural safety certificate in control systems;
	\item a subset of the \textsc{Cora} benchmark~\cite{KoLaAl24,BrBaJoWu24}, where disjuncts from the original queries where separated into different queries;
	\item \textsc{MaxLeaky}, a fresh benchmark that includes simple queries over a newly-trained DNN employing \emph{Maxpool} and \emph{LeakyRelu}~\cite{DuSiCh22} activation functions;
	\item a robotic navigation system benchmark~\cite{AmCoYeMaHaFaKa23}; and,
	\item the \textsc{SafeNLP} benchmark~\cite{CaArDaIsDiKiRiKo23,CaDiKoArDaIsKaRiLe25}.
\end{inparaenum}
These benchmarks were selected due to their relevance to
safety-critical applications their diversity in size and activation
functions, and their established use in previous work or at the DNN verification competition (VNNCOMP)~\cite{BrBaJoWu24,BrMuBaJoLi23}. Notably, we included all VNNCOMP benchmarks for which the proof-producing version of Marabou reasonably scales.

\subsection{Experimental Setup and Implementation}
Experiments were conducted on machines running Debian 12
Linux, each on a single CPU. We
divided the benchmarks into two sets according to network sizes, with
\textsc{Acas-Xu} and \textsc{Cora} consisting of larger DNNs, and the
remaining benchmarks consisting of smaller DNNs.
Memory and time configurations depended on category: for the larger
benchamrks we used 16GB RAM and a 5-hour \textsc{Timeout}, and for the
smaller ones --- 2GB RAM and a 2.5-hour \textsc{Timeout}.

Our implementation is based on the proof-producing version of the Marabou DNN verifier~\cite{KaHuIbJuLaLiShThWuZeDiKoBa19, WuIsZeTaDaKoReAmJuBaHuLaWuZhKoKaBa24,IsBaZhKa22}, into which we integrated our proof minimization algorithms and which we used as a baseline. The key metrics evaluated include \begin{inparaenum}[(i)]
	\item \emph{Proof size:} the number of proof vectors (for both lemmas and leaves) in the produced proofs of \unsat{} queries;
	\item \emph{Total Runtime:} the total runtime of the verification process, for all queries; and
    \item \emph{Total Proof-Checking time:}  the total proof-checking time for
          the \unsat{} queries, using the native Marabou proof checker.
\end{inparaenum} 
Note that a partial proof object is constructed during verification in
satisfiable queries as well (before concluding \sat); and thus, measuring verification time for all queries offers a more comprehensive comparison.

\subsection{Evaluation}
We begin by presenting a summary of the experimental results in~\Cref{table:evaluation_results}, which reports the average proof size generated during proof production, the average total verification time, and the average total proof-checking time, for each benchmark and proof minimization technique. In this table, averaging was conducted with respect to queries solved by all methods.
Furthermore, we provide a visual summary in~\Cref{fig:allres}, which
compares the proof sizes of all \unsat{} queries, and overall runtime across all queries.
We evaluated three variants of our method: \begin{inparaenum}[(i)]
	\item \emph{Dependency Analysis} (\textbf{Dep. Analysis}), as described in~\Cref{alg:lemmadeduction};
	\item \emph{Dependency Minimization} (\textbf{Dep. Min.}),
          which includes the implementation of~\Cref{alg:proofmin} on top of~\Cref{alg:lemmadeduction}; and
	\item \emph{Global Dependency Minimization} (\textbf{Dep. Min. Glob.}), which replaces~\Cref{alg:proofmin} with~\Cref{alg:globproofmin}.
\end{inparaenum}
For a deeper study of the effect on the total runtime, we consider
an additional running configuration, \textbf{Min. Baseline}, which
includes all data structures used for implementing our algorithms as
exemplified in~\Cref{sec:proofanalysis}, but without actually executing~\Cref{alg:lemmadeduction}. By running this version, we are able to separate the running time effect of the changes in data structures, from the effect of the analysis and minimization algorithms themselves.
All methods were implemented with the optimizations detailed in~\Cref{sec:misc}.

\begin{figure}[p]
	\centering
	\makebox[\textwidth][c]{
		\resizebox{1.1\textwidth}{!}{
			\begin{tabular}{m{0.3\textwidth} m{0.3\textwidth} m{0.3\textwidth}}
				\begin{minipage}[t]{\linewidth}\centering
					\textbf{Acas-Xu}\\
					\includegraphics[width=\linewidth]{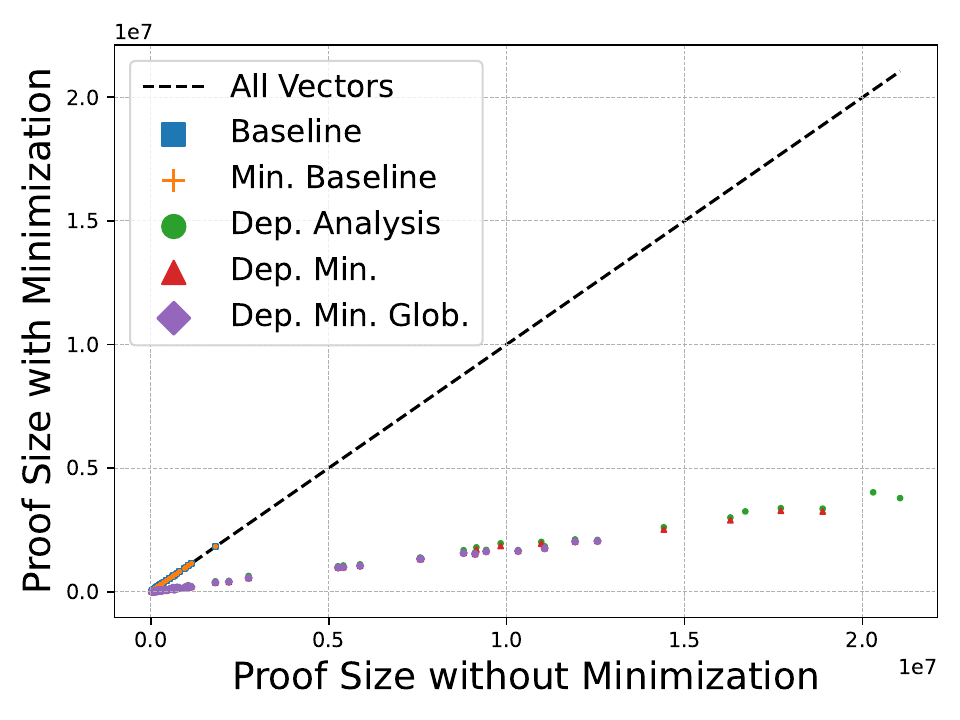}\\
					\includegraphics[width=\linewidth]{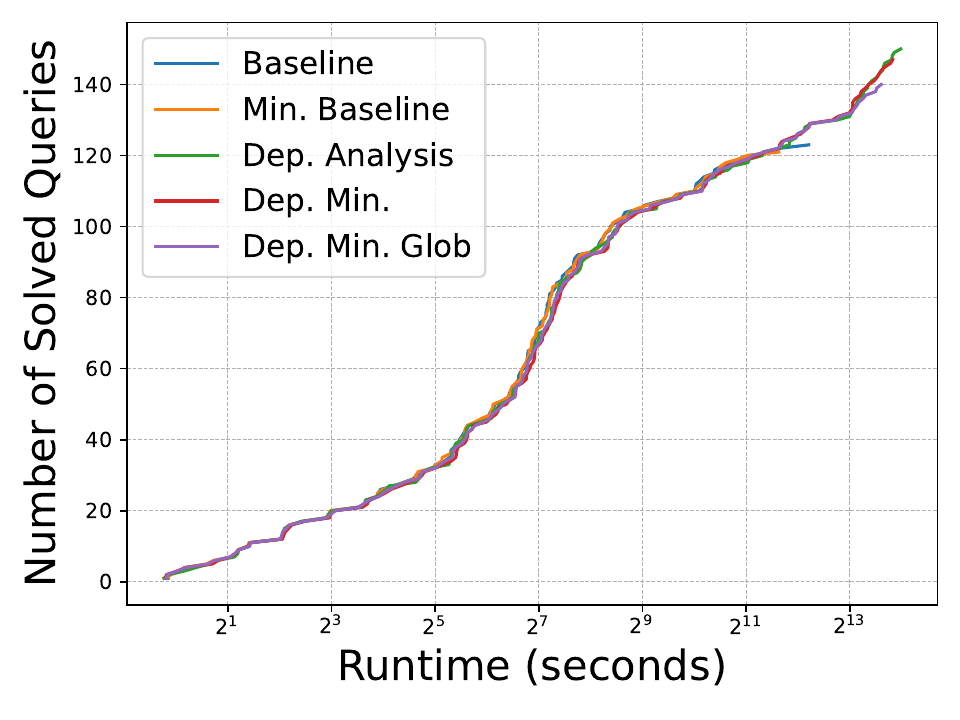}
				\end{minipage} &
				\begin{minipage}[t]{\linewidth}\centering
					\textbf{Cersyve}\\
					\includegraphics[width=\linewidth]{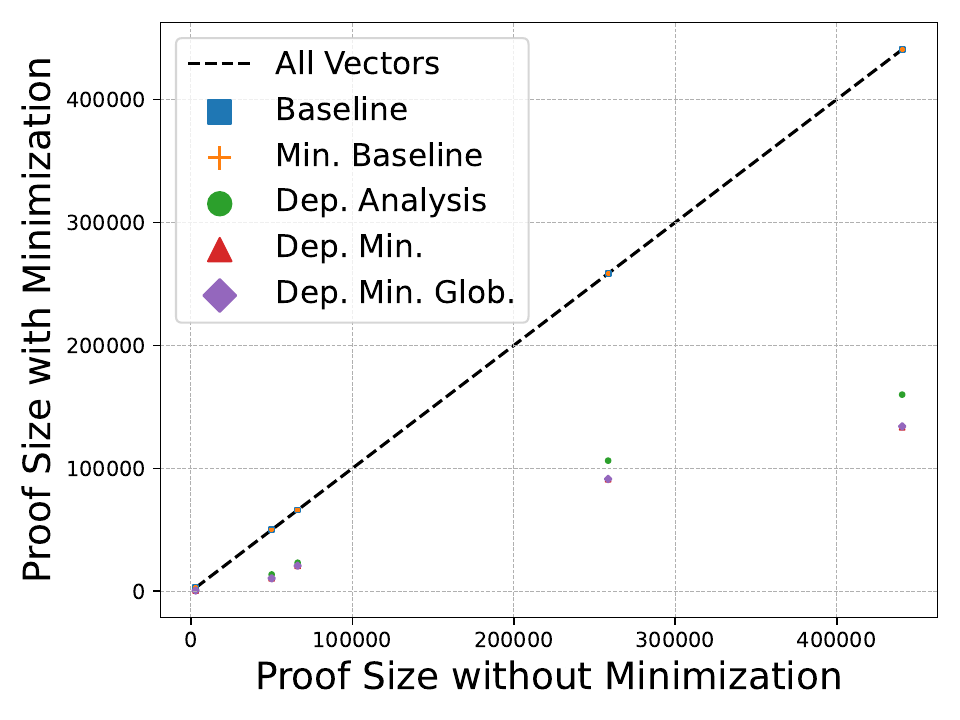}\\
					\includegraphics[width=\linewidth]{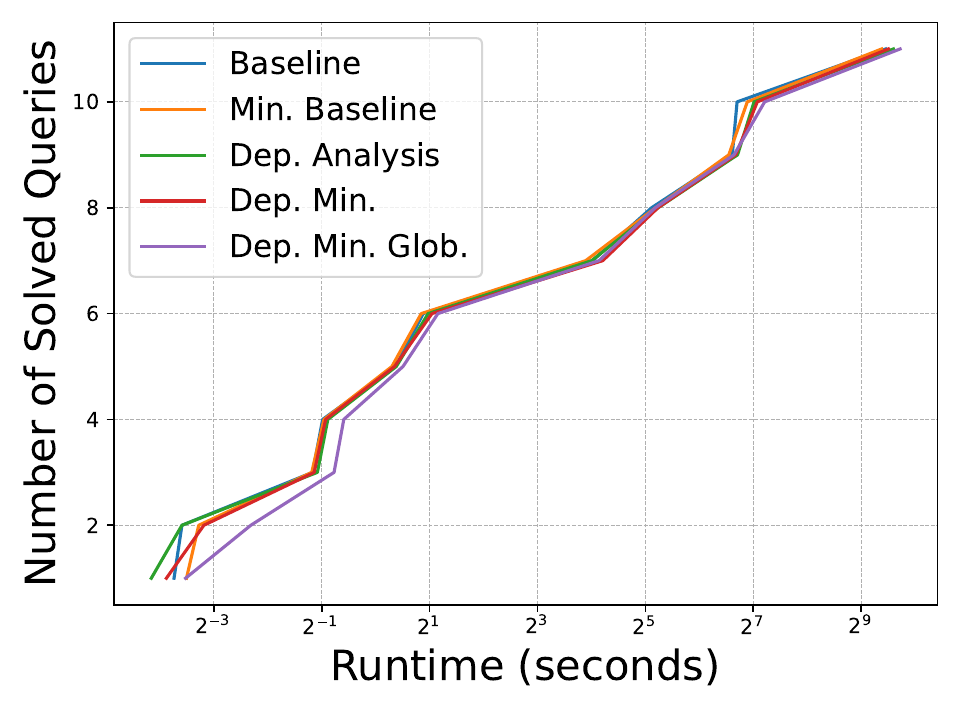}
				\end{minipage} &
				\begin{minipage}[t]{\linewidth}\centering
					\textbf{Cora}\\
					\includegraphics[width=\linewidth]{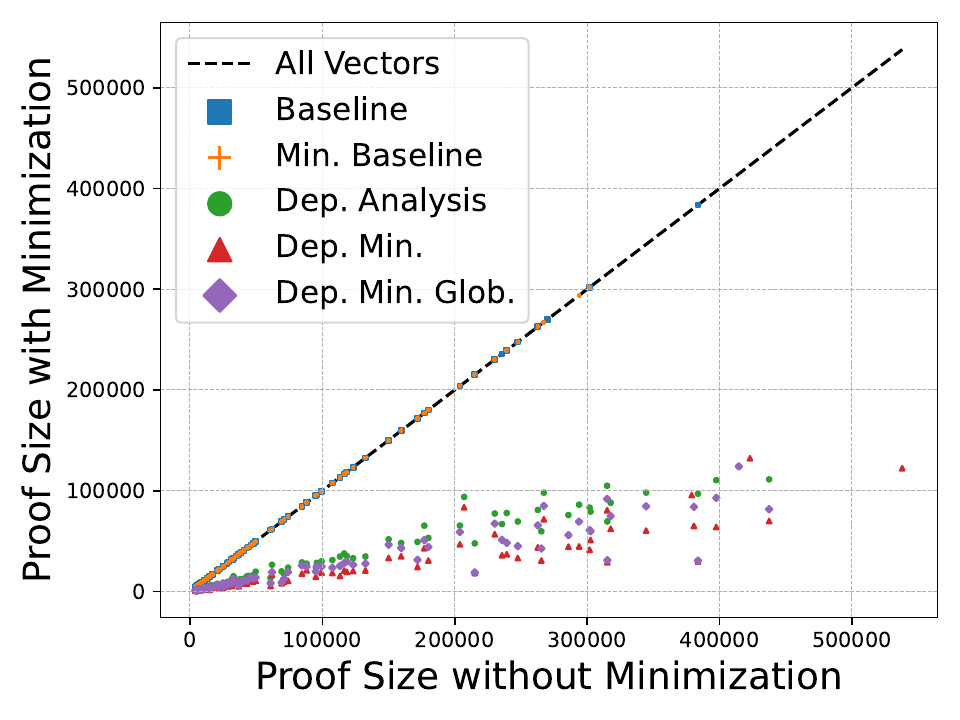}\\
					\includegraphics[width=\linewidth]{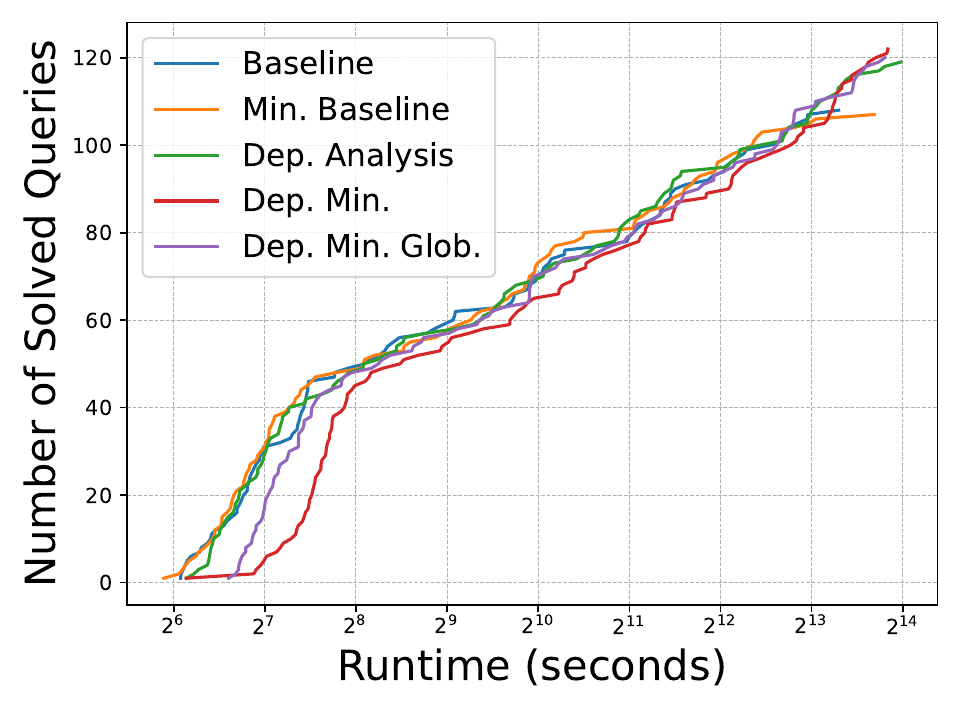}
				\end{minipage} \\
				\begin{minipage}[t]{\linewidth}\centering
					\textbf{Maxleaky}\\
					\includegraphics[width=\linewidth]{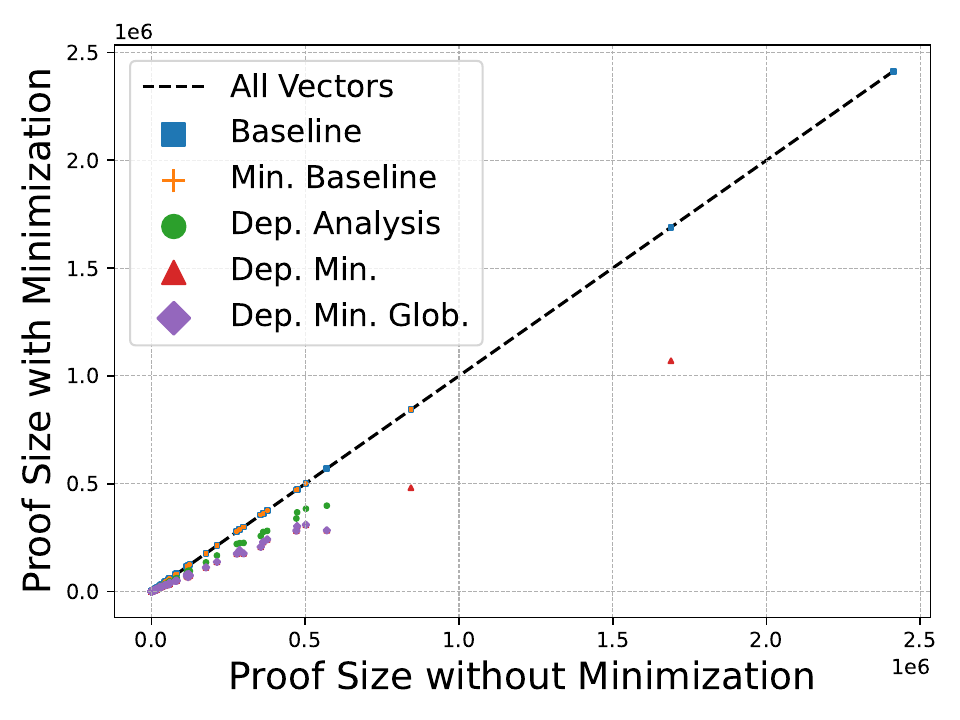}\\
					\includegraphics[width=\linewidth]{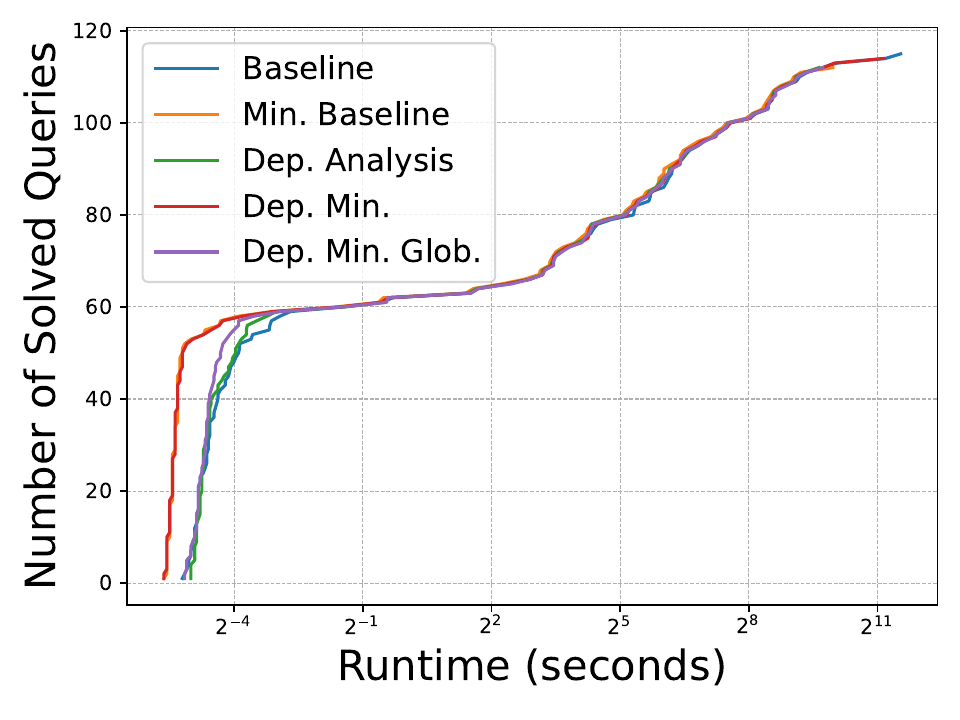}
				\end{minipage} &
				\begin{minipage}[t]{\linewidth}\centering
					\textbf{Robotics}\\
					\includegraphics[width=\linewidth]{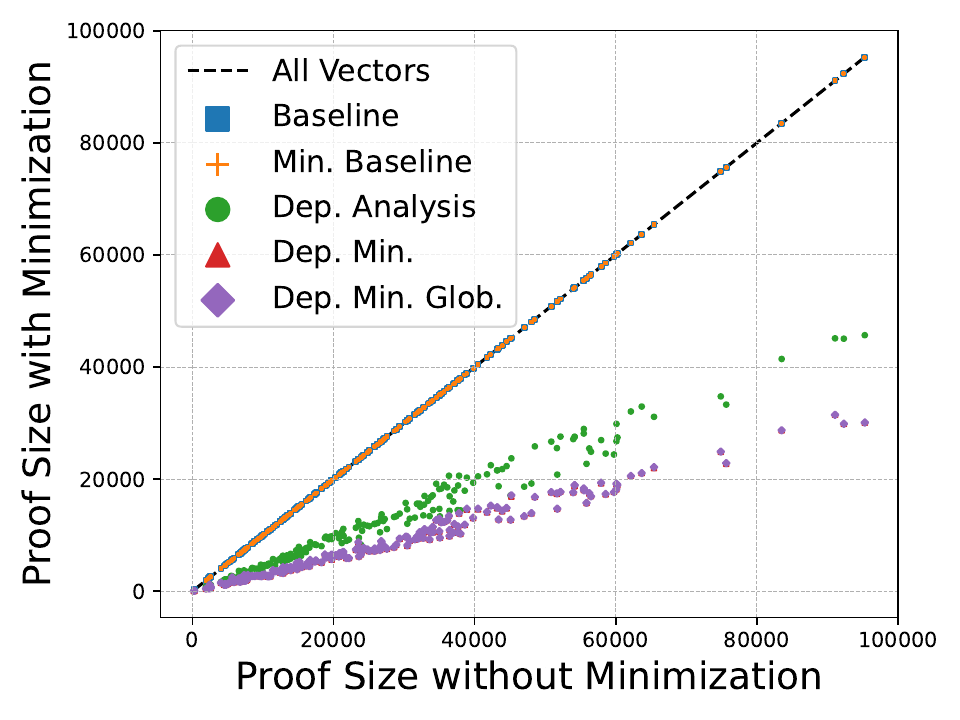}\\
					\includegraphics[width=\linewidth]{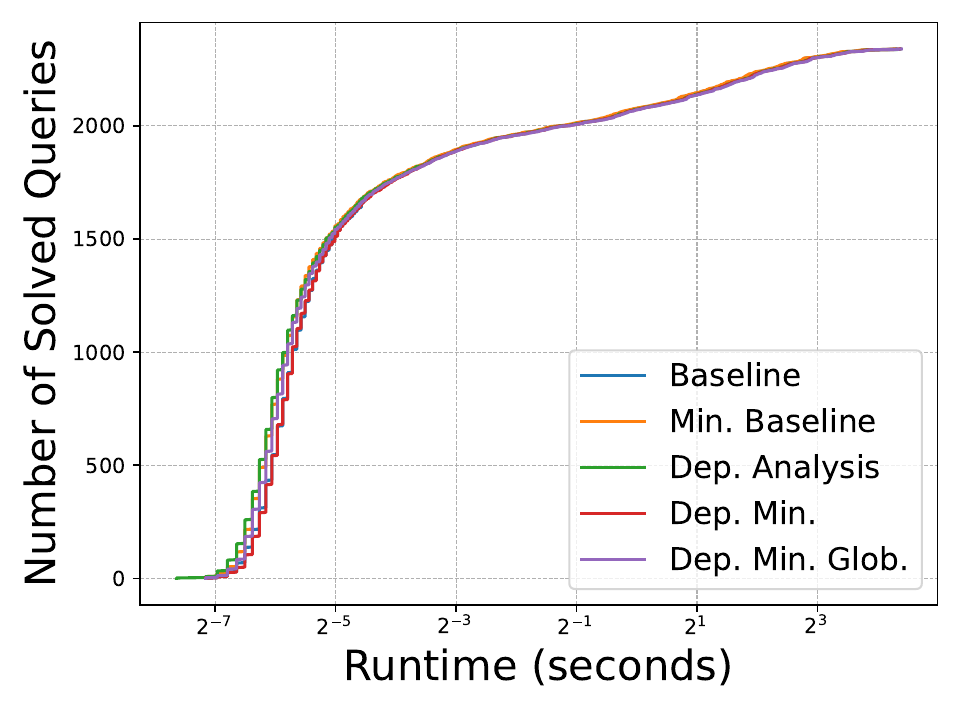}
				\end{minipage} &
				\begin{minipage}[t]{\linewidth}\centering
					\textbf{SafeNLP}\\
					\includegraphics[width=\linewidth]{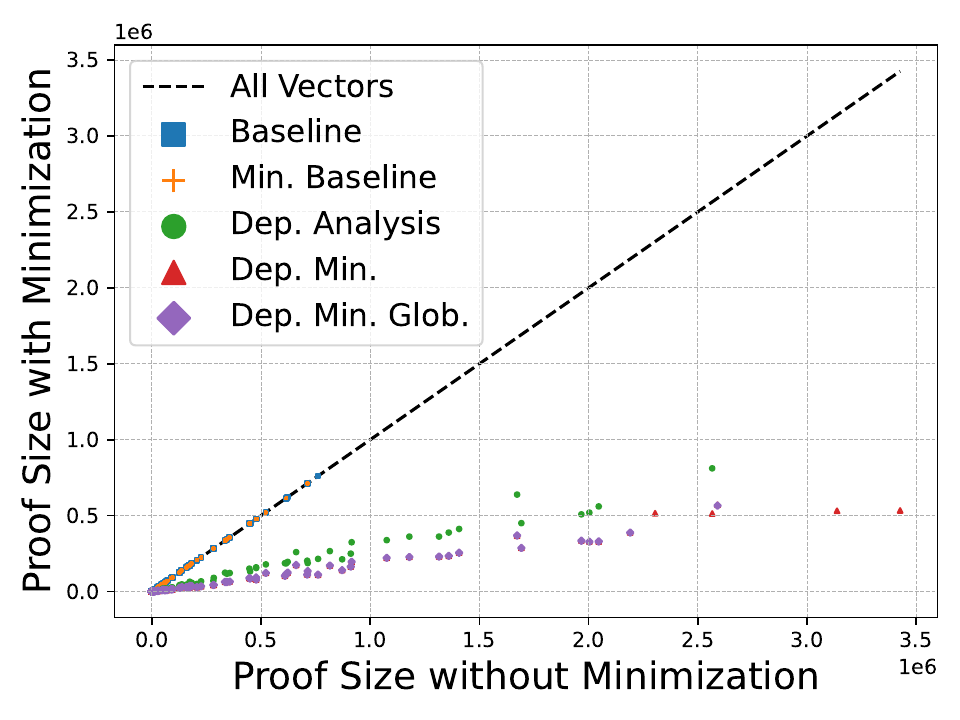}\\
					\includegraphics[width=\linewidth]{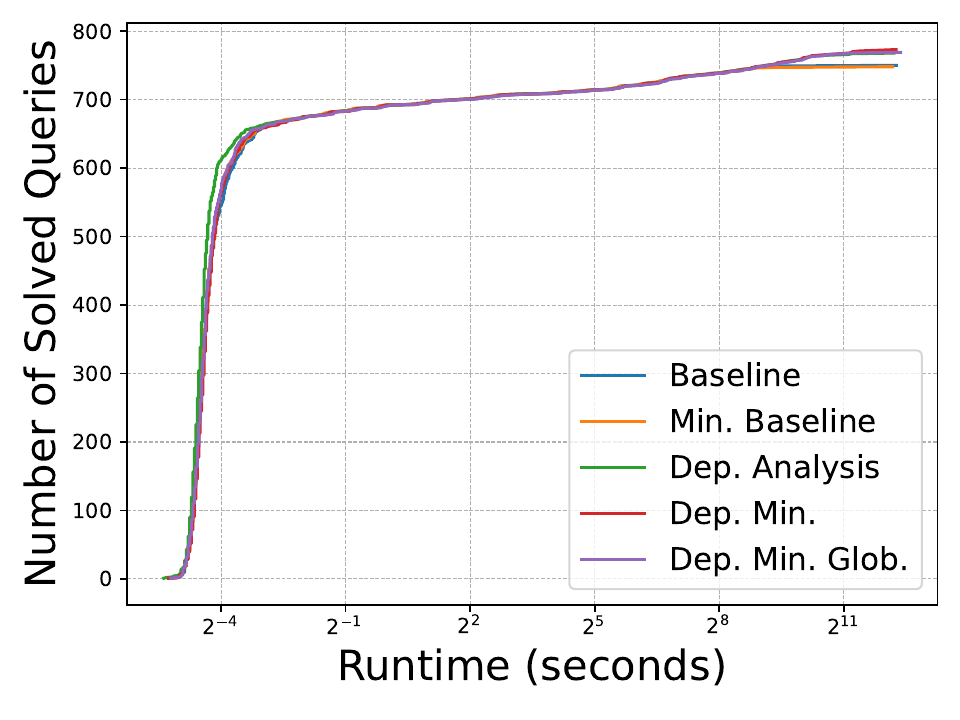}
				\end{minipage}
			\end{tabular}
		}
	}
	\caption{Proof size and runtime comparison across all benchmarks and methods.}
	\label{fig:allres}
\end{figure}

The results demonstrate a clear reduction in average proof size,
achieving a reduction rate of 37\%--82\%. The incurred verification
time overhead is 5\%--24\% for \textbf{Dep. Min.} and 7\%--20\% for \textbf{Dep. Min. Glob.}.
In addition, the results indicate a reduction in proof checking time
of up to 88\%. We observe that this reduction is not necessarily
correlated with the reduction rate of proof sizes; and we hypothesize
that this stems from the fact that checking time of each lemma is
dependent on its sparsity, and that our approach retains sparser
lemmas in the proof tree. 
Comparing the minimization algorithms, the results show
that~\Cref{alg:lemmadeduction} is responsible for most of the removed
lemmas, and both~\Cref{alg:proofmin} and~\Cref{alg:globproofmin}
improve it further with generally similar impact, except for the
\textsc{Cora} benchmarks. However, both algorithms incur a small
overhead comparing to using \Cref{alg:lemmadeduction} alone.
Analyzing \textbf{Min. Baseline} indicates that changes in data structures contributed to a minor improvement in both verification time and proof-checking time, for most benchmarks. This has moderated the time overhead of the minimization algorithms. As expected, bookeeping more information slightly increased the number of overall \textsc{Memout}s of \textbf{Min. Baseline} compared to \textbf{Baseline}. When applying the minimization algorithms however, the number of \textsc{Memout}s decreases significantly, due to the reduction of proof size (see subsections below).

In the subsequent subsections, we provide a more detailed analysis of the results for each benchmark suite.

\begin{table}[h!]
	\vspace{-0.8cm}
	\scriptsize
	\centering
	\caption{Average proof size (number of leaves and lemmas), verification time for all queries and proof-checking time for \unsat{} queries, using Marabou across benchmarks. Lower values indicate more compact proofs. Each non-baseline entry includes both reported numbers, and percentage of change with respect to baseline. Best results among our three algorithms are marked in \textbf{bold}.}
	\label{table:evaluation_results}
		\makebox[\textwidth][c]{
	\begin{tabular}{|c|c|c|c|c|c|}
		 \hline
		\textbf{Benchmark} & \multicolumn{5}{c|}{\makecell{\textbf{Avg. Proof Size [$\pm$\% of Baseline]} \\ \Xhline{0.1pt}\textbf{Avg. Verification Time (Sec.) [$\pm$\% of Baseline]} \\ \Xhline{0.1pt} \textbf{ Avg. Checking Time (Sec.) [$\pm$\% of Baseline]}}} \\
		\cline{2-6}
		& \textbf{Baseline\cite{IsBaZhKa22}} & \textbf{Min. Baseline} & \textbf{Dep. Analysis} & \textbf{Dep. Min.} &  \textbf{Dep. Min. Glob.} \\
		\hline
		\textbf{Acas-Xu} &
			\makecell{302021.14 \\ \\ \Xhline{0.1pt} 265.98 \\ \\ \Xhline{0.1pt} 66.37 \\ } &
			\makecell{302021.14 \\ 0\% \\ \Xhline{0.1pt} 269.42 \\ +1.2\% \\ \Xhline{0.1pt} 46.58 \\ -29.8\%} &
	 		\makecell{59231.07 \\ -80.4\% \\ \Xhline{0.1pt} 299.96 \\ +12.8\% \\ \Xhline{0.1pt} 8.67 \\ -86.9\% } &
	 	 	\makecell{\textbf{ 52890.17} \\ \textbf{-82.5\%}  \\ \Xhline{0.1pt} 294.58 \\ +10.7\% \\ \Xhline{0.1pt} \textbf{7.78} \\ \textbf{-88.3\% }} &
	 		\makecell{53144.84 \\ -82.4\% \\ \Xhline{0.1pt}\textbf{ 288.42} \\ \textbf{+8.4\%} \\ \Xhline{0.1pt} 8.15 \\ -87.7\% } \\ 
		\hline
				 \textbf{cersyve}& \makecell{163592.4\\ \\  \Xhline{0.1pt} 86.99 \\ \\ \Xhline{0.1pt} 13.77 \\ } & \makecell{163592.4 \\ 0\% \\\Xhline{0.1pt} 84.84 \\ -2.5\% \\\Xhline{0.1pt} 13.66 \\ -0.8\% } & \makecell{60764.6 \\ -62.9\% \\ \Xhline{0.1pt} 96.90 \\ +11.4\% \\ \Xhline{0.1pt} 3.14 \\ -77.2\%} & \makecell{\textbf{51283.6} \\ \textbf{-68.6\%}\\ \Xhline{0.1pt} \textbf{93.02} \\ \textbf{+6.9\% } \\ \Xhline{0.1pt} \textbf{2.62} \\ \textbf{-81\% }} & \makecell{51492.6 \\ -68.5\% \\ \Xhline{0.1pt} 104.57 \\ +20.2\%\\ \Xhline{0.1pt} 2.77 \\ -79.9\% } \\ 
		\hline
				 \textbf{Cora}& \makecell{49202.05\\ \\ \Xhline{0.1pt} 1277.88 \\\\ \Xhline{0.1pt} 1633.47 \\ } &
				  \makecell{49202.05 \\ 0\% \\\Xhline{0.1pt} 1231.5 \\ -3.6\% \\\Xhline{0.1pt} 1572.66 \\ -3.7\% } &
				   \makecell{15199.58 \\ -69.1\% \\ \Xhline{0.1pt} \textbf{1262.87} \\ \textbf{-1.2\%} \\ \Xhline{0.1pt} 455.88 \\ -72.1\% } & 
				   \makecell{\textbf{8819.75} \\ \textbf{-82.1\%} \\ \Xhline{0.1pt} 1583.02 \\ +23.9\% \\ \Xhline{0.1pt} \textbf{266.17} \\ \textbf{-83.7\%} } &
				   \makecell{11549.88 \\ -76.5\% \\ \Xhline{0.1pt} 1405.65 \\ +10\%\\ \Xhline{0.1pt} 365.32 \\ -77.6\%} \\ 
		\hline
				 \textbf{MaxLeaky}& 
				 \makecell{55342.88 \\ \\ \Xhline{0.1pt} 60.85 \\ \\ \Xhline{0.1pt} 2.67 \\ } & 
				 \makecell{55342.88 \\ 0\% \\\Xhline{0.1pt} 56.34 \\ -7.4\% \\\Xhline{0.1pt} 2.63  \\ -1.5\% } & 
				 \makecell{42224.96 \\ -23.7\% \\ \Xhline{0.1pt} \textbf{59.01} \\ \textbf{ -3\%}\\ \Xhline{0.1pt} 2.07 \\ -22.5\%} &  
				 \makecell{34852.52 \\ -37\% \\ \Xhline{0.1pt} 59.7 \\ -1.9\% \\ \Xhline{0.1pt} \textbf{1.89} \\ \textbf{-29.2\%}}& 
				 \makecell{\textbf{34847.12} \\ \textbf{-37\% }\\ \Xhline{0.1pt} 60.01 \\ -1.4\% \\ \Xhline{0.1pt} \textbf{1.89} \\ \textbf{-29.2\%}} \\ 
		\hline
				 \textbf{Robotics}& 
				 \makecell{25549.65\\ \\ \Xhline{0.1pt} 0.55\\ \\ \Xhline{0.1pt} 0.32\\ } &
				  \makecell{25549.65 \\ 0\% \\\Xhline{0.1pt} 0.53 \\ -3.6\%  \\\Xhline{0.1pt} 0.31 \\-3.1\%} & 
				  \makecell{11981.7 \\ -53.1\% \\ \Xhline{0.1pt} \textbf{0.57} \\ \textbf{ +3.6\%}\\ \Xhline{0.1pt} 0.15 \\ -53.1\%} &
				  \makecell{\textbf{8016.64} \\ \textbf{-68.6\%}\\ \Xhline{0.1pt} 0.58 \\ +5.4\%\\ \Xhline{0.1pt} \textbf{0.11} \\ \textbf{-65.6\%}} & 
				  \makecell{8026.99 \\ -68.6\% \\ \Xhline{0.1pt} 0.59 \\ +7.3\% \\ \Xhline{0.1pt} \textbf{0.11} \\ \textbf{-65.6\%}} \\ 
		\hline
				 \textbf{SafeNLP}& \makecell{55066.71\\\\ \Xhline{0.1pt} 14.69 \\\\ \Xhline{0.1pt} 5.89 \\ } &
				 \makecell{55066.71\\ 0\% \\\Xhline{0.1pt} 14.27 \\ -2.8\% \\\Xhline{0.1pt} 7.78 \\ +32\% } &
				 \makecell{17135.91 \\ -68.9\% \\ \Xhline{0.1pt} \textbf{14.73} \\ \textbf{+0.3\%}\\ \Xhline{0.1pt} 1.59 \\ -73\%} &
				 \makecell{10118.98 \\ -81.6\% \\ \Xhline{0.1pt} 14.94 \\ +1.7\% \\ \Xhline{0.1pt} \textbf{0.96} \\ \textbf{-83.7\%}} &  \makecell{\textbf{10062.95} \\ \textbf{ -81.7\%}\\ \Xhline{0.1pt} 15.88 \\ +8.1\%\\ \Xhline{0.1pt} 0.97 \\ -83.5\%}\\ 
		\hline
	
	\end{tabular}
	}
		\vspace{-0.4cm}
\end{table}
\subsubsection{\textsc{Acas-Xu} Benchmark.}
The \textsc{Acas-Xu} suite consists of fully connected feedforward networks
trained to support drone collision avoidance decision-making. We
ran experiments on a set of 45 networks, each comprised of 6 hidden
layers, each of which contained 50 neurons with \relu{} activation
functions. For each network we verified 4 different specifications~\cite{KaBaDiJuKo21}.
This results in a total of 180 queries.
The comparison in~\Cref{table:evaluation_results} was conducted on 82 \unsat{} and 39 \sat{} queries solved by all methods. Notably, a single query was either solved with a failure during proof certification, or terminated due to memory out. This result is consistent with the original evaluation of the \textsc{Acas-Xu} benchmark in~\cite{IsBaZhKa22}, and we hypothesize that for this particular verification instance, the proof-producing Marabou is numerically unstable.  

To assess robustness under resource constraints, we also report the
number of queries that were successfully solved within the 5-hour
timeout and 16GB memory limit. Table~\ref{table:acas_solved}
summarizes the number of solved queries for each method. As shown, all three
minimization methods outperform the baseline, solving more queries by avoiding memory blowup
during proof construction.

\begin{table}[h]
	\vspace{-0.8cm}
	\centering
	\caption{Number of \textsc{Acas-Xu} \sat{} queries solved, \unsat{} queries solved, timed out, and failed due to memory exhaustion, out of 180 queries.}
	\label{table:acas_solved}
	\begin{tabular}{|l|c|c|c|c|}
		\hline
		\textbf{Method} & \textbf{\#\sat{}} & \textbf{\#\unsat{}} & \textbf{\#Timeout} & \textbf{\#Memout} \\
		\hline
		Baseline & 41 & 82 & 0 & 57 \\
		\hline
		Min. Baseline & 39 & 82 & 0 & 59 \\
		\hline
		Dep. Analysis & 44 & 107 & 29 & 0 \\
		\hline
		Dep. Min. & 44 & 104 & 21 & 11 \\
		\hline
		Dep. Min. Glob. & 44 & 96 & 20 & 20 \\
		\hline
	\end{tabular}
	\vspace{-0.8cm}
\end{table}

\subsubsection{cersyve Benchmark.}
The \textsc{cersyve} suite consists of fully connected feedforward networks
serving as neural certificates for control systems. We
ran experiments on a set of 12 networks, with 65 to 198
\relu{} neurons. For each network we verified a single property, as in the VNNCOMP 2025 benchmark~\cite{VNNCOMP25}.
This results in a total of 12 queries, out of which 6 are \unsat. For this benchmark, all methods successfully verified 11 queries, whereas a single \unsat{} query fails due to memory restrictions when tested on all methods. 

\subsubsection{Cora Benchmark.}
The \textsc{Cora} benchmark consists of local adversarial robustness queries for image classifiers trained on several datasets. For our experiments, we used a subset of these queries on a feedforward DNN with seven layers of 250 \relu{} neurons each, trained on the MNIST dataset. To extend the benchmark and improve scalability, each query was further divided into nine sub-queries, each corresponding to a single possible misclassification option. In total, this yields 189 queries.
The comparison in~\Cref{table:evaluation_results} was conducted on 104 \unsat{} queries solved by all methods.

In~\Cref{table:cora_solved}, we report the number of instances verified under a time limit of 5 hours and a memory limit of 16GB. As previously observed, the analysis and minimization algorithms increase the total number of solved queries, this time primarily due to improvements in overall running time. Since the reduction in proof-checking time is more significant than the reduction in verification time (as shown in~\Cref{table:evaluation_results}), we conclude that the overall improvement is mainly attributable to the reduced time of proof-checking.

\begin{table}[h]
	\vspace{-0.8cm}
	\centering
	\caption{Number of \textsc{Cora} \sat{} queries solved, \unsat{} queries solved, timed out, and failed due to memory exhaustion, out of 189 queries.}
	\label{table:cora_solved}
	\begin{tabular}{|l|c|c|c|c|}
		\hline
		\textbf{Method} & \textbf{\#\sat{}} & \textbf{\#\unsat{}} & \textbf{\#Timeout} & \textbf{\#Memout} \\
		\hline
		Baseline & 1 & 107 & 81 & 0 \\
		\hline
		Min. Baseline & 1 & 106 & 81 & 1 \\
		\hline
		Dep. Analysis & 1 & 118 & 70 & 0 \\
		\hline
		Dep. Min. & 0 & 122 & 67 & 0 \\
		\hline
		Dep. Min. Glob. & 1 & 119 & 69 & 0 \\
		\hline
	\end{tabular}
	\vspace{-0.8cm}
\end{table}

\subsubsection{MaxLeaky Benchmark.}
The \textsc{MaxLeaky} benchmark includes a single DNN trained to classify wine types based on chemical composition~\cite{CoCeAnFeMaRe09}. To test our approach over a DNN which employs various activation functions,  we trained a fresh DNN with 64 \relu, 32 LeakyReLU, and 16 Maxpool neurons, to classify the dataset. We then generated a single adversarial robustness query for 133 train and 44 test data points. This results in a total of 177 queries.
The comparison in~\Cref{table:evaluation_results} was conducted on 99 \unsat{} and 12 \sat{} queries solved by all methods.

\Cref{table:maxleaky_solved} includes the number of instances verified under the limitation of 2.5-hour time and 2GB memory.
The results indicate a minor reduction in the number of solved queries when introducing the minimization algorithms. This result is unique, and is likely a result of the relatively moderate reduction of proof size (see~\Cref{table:evaluation_results}), combined with the small overhead introduced when keeping additional information for each lemma (see~\Cref{sec:proofanalysis}).

\begin{table}[h]
	\vspace{-0.8cm}
	\centering
	\caption{Number of \textsc{MaxLeaky} \sat{} queries solved, \unsat{} queries solved, timed out, and failed due to memory exhaustion, out of 177 queries.}
	\label{table:maxleaky_solved}
	\begin{tabular}{|l|c|c|c|c|}
		\hline
		\textbf{Method} & \textbf{\#\sat{}} & \textbf{\#\unsat{}} & \textbf{\#Timeout} & \textbf{\#Memout} \\
		\hline
		Baseline & 12 & 163 & 1 & 1 \\
		\hline
		Min. Baseline & 12 & 160 & 0 & 5 \\
		\hline
		Dep. Analysis & 12 & 160 & 1 & 4 \\
		\hline
		Dep. Min. & 12 & 162 & 0 & 3 \\
		\hline
		Dep. Min. Glob. & 12 & 160 & 0 & 5 \\
		\hline
	\end{tabular}
	\vspace{-0.8cm}
\end{table}

\subsubsection{Robotics Navigation Benchmark.}
The robotic navigation benchmark~\cite{AmCoYeMaHaFaKa23} involves DNN controllers for obstacle avoidance.
We tested our approach on 2340 queries from this benchmark, where 2126 of them are \sat{} queries, and 214 are \unsat{}. Each query
includes a DNN with the \relu{} activation only. The DNN is comprised of an input layer with 9 neurons, two hidden layers with 16 neurons each, and an output layer of 3 neurons. 
For this benchmark, all methods were able to verify all queries and produce proofs in the time and memory limitations. 

\subsubsection{SafeNLP Benchmark.}
The \textsc{safeNLP} benchmark~\cite{CaArDaIsDiKiRiKo23,CaDiKoArDaIsKaRiLe25} examines robustness to word- and sentence-level perturbations in the embedding space of several sentences over two fully-connected DNNs. Both DNNs have the same architecture, which is comprised of an input layer of 30 nodes, a single hidden layer of 128 nodes with ReLU activation functions, and an output layer of 2 nodes. The benchmark consists of 1080 queries in total. The comparison in~\Cref{table:evaluation_results} was conducted on 149 \unsat{} and 599 \sat{} queries solved by all methods.

Table~\ref{table:safenlp_solved} reports the number of queries successfully verified within the resource limits (2.5 hours and 2GB RAM). Here again, our minimization methods demonstrate improved performance, reducing memory exhaustion and increasing the number of successful verification queries compared to the baseline.

\begin{table}[h]
	\vspace{-0.8cm}
	\centering
	\caption{Number of safeNLP \sat{} queries solved, \unsat{} queries solved, timed out, and failed due to memory exhaustion, out of 1080 queries.}
	\label{table:safenlp_solved}
	\begin{tabular}{|l|c|c|c|c|}
	\hline
	\textbf{Method} & \textbf{\#\sat{}} & \textbf{\#\unsat{}} & \textbf{\#Timeout} & \textbf{\#Memout} \\
	\hline
	Baseline & 599 & 196 & 13 & 272 \\
	\hline
	Min. Baseline & 599 & 194 & 7 & 280 \\
	\hline
	Dep. Analysis & 599 & 214 & 36 & 231 \\
	\hline
	Dep. Min. & 599 & 219 & 54 & 208 \\
	\hline
	Dep. Min. Glob. & 599 & 215 & 42 & 224 \\
	\hline
\end{tabular}
\vspace{-0.8cm}
\end{table}

\section{Related Work}
\label{sec:relwork}
Producing proofs as witnesses of correctness is a common practice in both SMT and SAT solving. In both communities, the proof size tends to be significant, limiting their applicability~\cite{BaDeFo15, HeBi15,SoBi09, BaBaCoDuKrLaNiNoOzPrReTiZo23}. Our work builds on concepts inspired by \emph{lazy proofs}~\cite{KaBaTiReHa16}, where  proof objects are constructed gradually, using only necessary lemmas. 
In DNN verification, proof production has  been studied only recently, and in a handful of works~\cite{IsBaZhKa22, SiSaMeSi25}. The authors of~\cite{SiSaMeSi25} consider a symbolic approach for increasing the reliability of DNN verifiers, which reduces the problem to an additional SMT query. Our approach is then not directly applicable to~\cite{SiSaMeSi25}. 

The study of dependencies analysis of \unsat{} results has been
explored in various contexts~\cite{Ju01, Ju04, BrMa08, GuStTr16,
	LiSa08} often using the name \emph{\unsat{} core}. To the best of our knowledge, ours is the first to
investigate this in the context of DNN verification.
In DNN verification, previous studies have attempted to analyze dependencies between the
phases of neurons as part of the DNN verification
process~\cite{BoKoPaKrLo20}. As fixing neuron phases could be captured by proof lemmas, it would be interesting to study the synergies between that approach and ours.

In linear programming, the well-established concept of \emph{Irreducible Infeasible Systems}~\cite{ChDr91} is defined as a minimal set of constraints that maintain unsatisfiablity. The method to detect IISs from~\cite{ChDr91} has been used in DNN verifiers~\cite{Eh17}, and is implemented in the Gurobi linear optimizer~\cite{gurobi} often used in DNN verifiers.
Although  IIS detection algorithms typically rely on repeated LP solver invocations, which can be computationally expensive, it presents a compelling direction for future exploration in this work's context.

\section{Conclusion and Future Work}
\label{sec:conclusion}
In this work, we developed methods to reduce the size of \unsat{} proofs generated by DNN verifiers, as presented in~\cite{IsBaZhKa22}. These proofs can be checked by a trusted, external proof checker in order to improve reliability of the DNN verifier. Importantly, DNN verifiers currently assume ideal mathematical implementation of the DNN itself, and modeling any specific DNN implementations remains a challenge for the DNN verification community in general~\cite{CoDaGiIsJoKaKoLeMaSiWu25}. 
Our goal here is to improve memory consumption in proof-producing verifiers and reduce proof checking time, thereby enhancing the practicality of reliable DNN verification. 

To this end, we designed an algorithm that analyzes, for each lemma in the proof, its dependencies on previously derived lemmas. This enables us to eliminate lemmas that are eagerly learned but are ultimately irrelevant for proving \unsat. In addition, we designed an algorithm to minimize the number of such dependencies, with two variants. All algorithms have been implemented on top of the proof-producing version of the Marabou DNN verifier~\cite{WuIsZeTaDaKoReAmJuBaHuLaWuZhKoKaBa24}.

Our results show that~\Cref{alg:lemmadeduction} reduces proof size by an order of magnitude across multiple benchmarks. The dependency minimization algorithms further reduces proof size, though more moderately. When comparing the two variants of the minimization algorithm, both versions achieve similar results on almost all benchmarks.
Moreover, the performance overhead introduced by our methods is reasonable, especially when considering improvements in proof checking time. Notably, the time overhead from the dependency analysis is reduced due to the use of some alternative data structures.

\mysubsection{Future Work.}
Moving forward, we intend to continue this work along several paths.
First, we aim to explore different variants of proof minimization. In
particular, our current algorithms consider a fixed proof vector $w$,
which we could try to alter in order to reduce the number of dependencies even further. Additionally, it would be interesting to explore additional alternatives to the minimization algorithm, attempting to remove dependencies according to other heuristics, e.g., by their stage of deduction. 
Furthermore, it would also be interesting to consider the sparsity of
the proof vectors and prioritize removing dense vectors, as they
(heuristically) depend on many other bound lemmas.

Second, we aim to extend the proof checker as
in~\cite{DeIsKoStPaKa25}, to support checking proofs of bound
tightening lemmas. We will then evaluate the reduction of
proof-checking time for this version of the checker as well.

A third direction for future research is extending our work to support additional subroutines commonly used in DNN verifiers, particularly abstract interpretation. This integration is inherently more complex, as it requires proving linear relaxations of non-linear constraints. However, once this support is established, we believe our proof minimization approach and technique could also be useful in this context.

Fourth, our work's analysis method paves the way for a \emph{conflict analysis} method, as a part of \emph{Conflict-Driven Clause Learning}\cite{SiSk96,SiSk99,BaSc97} scheme, which is commonly used in SAT and SMT solving~\cite{BaTi18,BiFaFaFlFrPo24}. Conceptually, in search-based solving, CDCL
identifies subspaces of the search space which are ``similar'' to
subspaces already traversed, and which were shown not to contain any
satisfying assignment. This
similarity guarantees that the new subspaces also do not contain any
such assignment, and they can be safely skipped by the
verifier.  A core component in CDCL is the generation of
\emph{conflict clauses}, a representation of an \unsat{} subspaces, that guide
the solver not to search in other subspaces that are also guaranteed
to be \unsat{}.   Naturally, the algorithm for deriving conflict clauses is a
key to a successful CDCL framework. Our work can serve as a basis for the derivation of potentially useful conflict clauses for DNN verifiers --- leading to a significant improvement in their performance.  Although several DNN verifiers employ CDCL or similar algorithms~\cite{LiYaZhHu24,DuNgDw23,Eh17,ZhBrHaZh24}, to the best of our knowledge, none of these  use \unsat{} proofs for deriving conflict clauses.

\mysubsection{Data Availability.} The implementation and scripts used for the experiments in~\Cref{sec:Evaluation} are publicly available~\cite{ReIs25}. 

\mysubsection{Acknowledgments.} 
The work of Isac, Refaeli and Katz was supported by the Binational
Science Foundation (grant numbers 2020250 and 2021769), the Israeli Science Foundation (grant number 558/24) and the European Union (ERC, VeriDeL, 101112713). Views and opinions expressed are however those of the author(s) only and do not necessarily reflect those of the European Union or the European Research Council Executive Agency. Neither the European Union nor the granting authority can be held responsible for them.

The work of Barrett was supported in part by the Binational Science Foundation (grant number 2020250), the National Science Foundation (grant numbers 1814369 and 2211505), and the Stanford Center for AI Safety.

\bibliographystyle{abbrv}
\bibliography{proofmin}

@article{KaBaDiJuKo21,
 author = {Katz, G. and Barrett, C. and Dill, D. and Julian, K. and Kochenderfer, M.},
 title = {{Reluplex: a Calculus for Reasoning about Deep Neural Networks}},
 journal = {Formal Methods in System Design (FMSD)},
 year = {2021},
}

@inproceedings{JiRi21,
	title = {{Exploiting Verified Neural Networks via Floating Point Numerical Error}},
	author = {Jia, K. and Rinard, M.},
	year = {2021},
	booktitle = {Proc. 28th Int. Static Analysis Symposium (SAS)},
	pages={191--205},
}

@inproceedings{KaHuIbJuLaLiShThWuZeDiKoBa19,
  	title = {{The Marabou Framework for Verification and Analysis of Deep Neural Networks}},
	author = {Katz, G. and Huang, D. and Ibeling, D. and Julian, K. and Lazarus, C. and Lim, R. and Shah, P. and Thakoor, S. and Wu, H. and Zelji\'c, A. and Dill, D. and Kochenderfer, M. and Barrett, C.},
	year = {2019},
	booktitle = {Proc. 31st Int. Conf. on Computer Aided Verification (CAV)},
	pages = {443--452}
}

@article{JuKoOw19,
	title = {{Deep Neural Network Compression for Aircraft Collision Avoidance Systems}},
	volume = {42},
	pages = {598--608},
	number = {3},
	journal = {Journal of Guidance, Control, and Dynamics},
	author = {Julian, K. and Kochenderfer, M. and Owen, M.},
	year = {2019},
}

@incollection{BaTi18,
  author                   = {Barrett, C. and Tinelli, C.},
  title                    = {{Satisfiability Modulo Theories}},
  editor                   = {Clarke, E. and Henzinger, T. and Veith, H. and Bloem, R.},
  booktitle                = {Handbook of Model Checking},
  year                     = {2018},
  pages                    = {305--343},
  publisher                = {Springer International Publishing},
}

@Misc{BoDeDwFiFlGoJaMoMuZhZhZhZi16,
  Title = {{End to End Learning for Self-Driving Cars}},
  Author = {Bojarski, M. and Del Testa, D. and Dworakowski, D. and Firner, B. and Flepp, B. and
            Goyal, P. and Jackel, L. and Monfort, M. and Muller, U. and Zhang, J. and Zhang, X. and
            Zhao, J. and Zieba, K.},
  Note = {Technical Report. \url{http://arxiv.org/abs/1604.07316}},
  Year = {2016}
}

@Misc{SzZaSuBrErGoFe13,
	Title = {{Intriguing Properties of Neural Networks}},
	Author = {Szegedy, C. and Zaremba, W. and Sutskever, I. and Bruna, J. and Erhan, D. and Goodfellow, I. and Fergus, R.},
	Note = {Technical Report. \url{http://arxiv.org/abs/1312.6199}},
	Year = {2013}
}

@misc{gurobi,
        key={Gurobi},
	title={{The Gurobi Optimizer}},
	Note={\url{https://www.gurobi.com/}},
}

@article{DeBj11,
	author = {de Moura, L. and Bj\o{}rner, N.},
	title = {{Satisfiability Modulo Theories: Introduction and Applications}},
	year = {2011},
	volume = {54},
	number = {9},
	journal = {Communications of the ACM},
	pages = {69--77},
}

@inproceedings{GeMiDrTsChVe18,
  Title                    = {{AI2: Safety and Robustness Certification of Neural Networks with Abstract Interpretation}},
  Author                   = {Gehr, T. and Mirman, M. and Drachsler-Cohen, D. and Tsankov, E. and Chaudhuri, S. and Vechev, M.},
  year                     = {2018},
  booktitle                = {Proc. 39th IEEE Symposium on Security and Privacy (S\&P)},
}

@Misc{TjXiTe17,
  Title = {{Evaluating Robustness of Neural Networks with Mixed Integer Programming}},
  Author = {Tjeng, V. and Xiao, K. and Tedrake, R.},
  Note = {Technical Report. \url{http://arxiv.org/abs/1711.07356}},
  Year = {2017}
}

@inproceedings{WaPeWhYaJa18,
  author    = {Wang, S. and Pei, K. and Whitehouse, J. and Yang, J. and Jana, S.},
  title     = {{Formal Security Analysis of Neural Networks using Symbolic Intervals}},
  booktitle = {Proc. 27th {USENIX} Security Symposium},
  pages     = {1599--1614},
  year      = {2018}
}

@inproceedings{HuKwWaWu17,
  Title                    = {{Safety Verification of Deep Neural Networks}},
  Author                   = {Huang, X. and Kwiatkowska, M. and Wang, S. and Wu, M.},
  booktitle                = {Proc. 29th Int. Conf. on Computer Aided Verification (CAV)},
  Year                     = {2017},
  Pages                    = {3--29}
}

@InProceedings{TrBaXiJo20,
  Title = {{Verification of Deep Convolutional Neural Networks Using ImageStars}},
  Author = {Tran, H.-D. and Bak, S. and Xiang, W. and Johnson, T.},
  Booktitle = {Proc. 32nd Int. Conf. on Computer Aided Verification (CAV)},
  Year = {2020},
  pages = {18--42}
}

@inproceedings{LyKoKoWoLiDa20,
  title = {{Fastened Crown: Tightened Neural Network Robustness Certificates}},
  author = {Lyu, Z. and Ko, C.-Y. and Kong, Z. and Wong, N. and Lin, D. and Daniel, L.},
  year = {2020},
  booktitle = {Proc. 34th AAAI Conf. on Artificial Intelligence (AAAI)},
  pages = {5037--5044},
}

@InProceedings{Eh17,
  Title = {{Formal Verification of Piece-Wise Linear Feed-Forward Neural Networks}},
  Author = {Ehlers, R.},
  Booktitle = {Proc. 15th Int. Symp. on Automated Technology for Verification and Analysis (ATVA)},
  pages = {269--286},
  Year = {2017},
}

@Article{BaDeFo15,
  Title                    = {{Proofs in Satisfiability Modulo Theories}},
  Author                   = {Barrett, Clark and de Moura, Leonardo and Fontaine, Pascal},
  Journal                  = {All about Proofs, Proofs for All},
  Year                     = {2015},
  Number                   = {1},
  Pages                    = {23--44},
  Volume                   = {55}
}

@InProceedings{KaBaTiReHa16,
  Title                    = {{Lazy Proofs for DPLL(T)-Based SMT Solvers}},
  Author                   = {Katz, G. and Barrett, C. and Tinelli, C. and Reynolds, A. and Hadarean, L.},
  Booktitle                = {Proc. 16th Int. Conf. on Formal Methods in Computer-Aided Design (FMCAD)},
  Year                     = {2016},
  Pages                    = {93--100}
}

@book{ChvatalLP,
  title={{Linear Programming}},
  author={V\'{a}clav Chv\'{a}tal},
  publisher={W. H. Freeman and Company},
  year={1983}
}

@article{BrMuBaJoLi23,
	title={{First Three Years of the International Verification of Neural Networks Competition (VNN-COMP)}},
	author={Brix, Christopher and M{\"u}ller, Mark and Bak, Stanley and Johnson, Taylor and Liu, Changliu},
	journal={Int. Journal on Software Tools for Technology Transfer},
	year={2023},
	pages={1-11},
}

@inproceedings{ElElIsDuGaPoBoCoKa24,
	title={{Robustness Assessment of a Runway Object Classifier for Safe Aircraft Taxiing}},
	author={Elboher, Yizhak and Elsaleh, Raya and Isac, Omri and Ducoffe, M{\'e}lanie and Galametz, Audrey and Pov{\'e}da, Guillaume and Boumazouza, Ryma and Cohen, No{\'e}mie and Katz, Guy},
	booktitle = {Proc. 43rd Int. Digital Avionics Systems Conf. (DASC)},
	year={2024}
}

@inproceedings{KoLeEdChMaLo23,
	title={{Verification of Semantic Key Point Detection for Aircraft Pose Estimation}},
	author={Kouvaros, Panagiotis and Leofante, Francesco and Edwards, Blake and Chung, Calvin and Margineantu, Dragos and Lomuscio, Alessio},
	booktitle={Proc. 20th  Int. Conf. on Principles of Knowledge Representation and Reasoning (KR)},
	pages={757--762},
	year={2023}
}

@inproceedings{AmFrKaMaRe23,
	Title = {{veriFIRE: Verifying an Industrial, Learning-Based Wildfire 
	Detection System}},
	Author = {Amir, G. and Freund, Z. and Katz, G. and Mandelbaum, E. and 
	Refaeli, I.},
	Booktitle = {Proc. 25th Int. Symposium on Formal Methods (FM)},
	Year = {2023},
	Pages = {648--656},
}

@InProceedings{IsBaZhKa22,
	Title = {{Neural Network Verification with Proof Production}},
	Author = {Isac, O. and Barrett, C. and Zhang, M. and Katz, G.},
	Booktitle = {Proc. 22nd Int. Conf. on Formal Methods in Computer-Aided Design
	(FMCAD)},
	Pages = {38--48}, 
	Year = {2022}
}

@book{GoBeCu16,
	title={{Deep Learning}},
	author={Goodfellow, I. and Bengio, Y. and Courville, A.},
	publisher={MIT Press},
	year={2016}
}

@inproceedings{IsZoBaKa23,
	title={{DNN Verification, Reachability, and the Exponential Function Problem}},
	author={Isac, Omri and Zohar, Yoni and Barrett, Clark and Katz, Guy},
	booktitle={Proc. 34th Int. Conf. on Concurrency Theory (CONCUR)},
	year={2023}
}

@book{Da63,
	title = {{Linear Programming and Extensions}},
	publisher = {Princeton University Press},
	author = {Dantzig, G.},
	year = {1963},
}

@misc{DuNgDw23,
	title={{A DPLL(T) Framework for Verifying Deep Neural Networks}},
	author={Duong, Hai and Nguyen, ThanhVu and Dwyer, Matthew},
	note   = {Technical Report. \url{http://arxiv.org/abs/2307.10266}},
	year={2023}
}

@inproceedings{LiYaZhHu24,
	title={{DeepCDCL: A CDCL-based Neural Network Verification Framework}},
	author={Liu, Zongxin and Yang, Pengfei and Zhang, Lijun and Huang, Xiaowei},
	booktitle={Proc. 18th Int. Symposium on Theoretical Aspects of Software Engineering (TASE)},
	pages={343--355},
	year={2024},
}

@inproceedings{BoKoPaKrLo20,
	title={{Efficient Verification of ReLU-Based Neural Networks via Dependency Analysis}},
	author={Botoeva, Elena and Kouvaros, Panagiotis and Kronqvist, Jan and Lomuscio, Alessio and Misener, Ruth},
	booktitle={Proc. 34th AAAI Conf. on Artificial Intelligence (AAAI)},
	pages={3291--3299},
	year={2020}
}

@inproceedings{BiFaFaFlFrPo24,
	title={{CaDiCaL 2.0}},
	author={Biere, Armin and Faller, Tobias and Fazekas, Katalin and Fleury, Mathias and Froleyks, Nils and Pollitt, Florian},
	booktitle={Proc. 36th Int. Conf. on Computer Aided Verification (CAV)},
	pages={133--152},
	year={2024}
}

@inproceedings{WuIsZeTaDaKoReAmJuBaHuLaWuZhKoKaBa24,
	title={{Marabou 2.0: A Versatile Formal Analyzer of Neural Networks}},
	author={H. Wu and O. Isac and A. Zelji\'c and T. Tagomori and M. Daggitt and W. Kokke and I. Refaeli and G. Amir and K. Julian and S. Bassan and P. Huang and O. Lahav and M. Wu and M. Zhang and E. Komendantskaya and G. Katz and C. Barrett},
	booktitle = {Proc. 36th Int. Conf. on Computer Aided Verification (CAV)},
	year={2024}
}

@inproceedings{SiSk96,
	title={{GRASP-A New Search Algorithm for Satisfiability}},
	author={Silva, J and Sakallah, Karem},
	booktitle={Proc. 15th  Int. Conf. on Computer Aided Design (ICCAD)},
	pages={220--227},
	year={1996}
}

@article{SiSk99,
	title={{GRASP: A Search Algorithm for Propositional Satisfiability}},
	author={Silva, JP Marques and Sakallah, Karem A},
	journal={IEEE Transactions on Computers},
	volume={48},
	number={5},
	pages={506--521},
	year={1999},
	publisher={IEEE}
}

@inproceedings{BaSc97,
	title={{Using CSP Look-Back Techniques to Solve Real-World SAT Instances}},
	author={Bayardo Jr, Roberto and Schrag, Robert},
	booktitle={Proc. 14th Nat. Conf. on Artificial Intelligence (AAAI)},
	pages={203--208},
	year={1997}
}

@InProceedings{GaGePuVe19,
	Title                    = {{An Abstract Domain for Certifying Neural Networks}},
	Author                   = {Singh, G. and Gehr, T. and Puschel, M. and Vechev, M.},
	Booktitle                = {Proc. 46th ACM SIGPLAN Symposium on Principles of Programming Languages (POPL)},
	Year                     = {2019},
}

@inproceedings{WaZhXuLiJaHsKo21,
	title={{Beta-Crown: Efficient Bound Propagation with Per-Neuron Split Constraints for Neural Network Robustness Verification}},
	author={Wang, Shiqi and Zhang, Huan and Xu, Kaidi and Lin, Xue and Jana, Suman and Hsieh, Cho-Jui and Kolter, Zico},
	booktitle={Proc. 35th Conf. on Neural Information Processing Systems (NeurIPS},
	year={2021}
}

@article{RaChBaTo22,
	title={{AI in Health and Medicine}},
	author={Rajpurkar, Pranav and Chen, Emma and Banerjee, Oishi and Topol, Eric J},
	journal={Nature Medicine},
	volume={28},
	number={1},
	pages={31--38},
	year={2022},
	publisher={Nature Publishing Group US New York}
}

@misc{Op22,
	title={{ChatGPT}},
	author={OpenAI},
	Note={\url{https://chatgpt.com}}
}

@inproceedings{Ju01,
	title={{Quickxplain: Conflict Detection for Arbitrary Constraint Propagation Algorithms}},
	author={Junker, Ulrich},
	booktitle={Proc. Workshop on Modelling and Solving Problems with Constraints},
	year={2001}
}

@inproceedings{Ju04,
	title={{Quickxplain: Preferred Explanations and Relaxations for Over-Constrained Problems}},
	author={Junker, Ulrich},
	booktitle={Proc. 19th Nat. Conf. on Artifical Intelligence (AAAI)},
	pages={167--172},
	year={2004}
}

@article{BrMa08,
	title={{Property-Directed Incremental Invariant Generation}},
	author={Bradley, Aaron and Manna, Zohar},
	journal={Formal Aspects of Computing},
	volume={20},
	pages={379--405},
	year={2008},
}

@misc{BrBaJoWu24,
	title={{The Fifth International Verification of Neural Networks Competition (VNN-COMP 2024): Summary and Results}},
	author={Brix, Christopher and Bak, Stanley and Johnson, Taylor and Wu, Haoze},
	Note   = {Technical Report. \url{http://arxiv.org/abs/2412.19985}},
	year={2024}
}

@inproceedings{AmCoYeMaHaFaKa23,
	title={{Verifying Learning-Based Robotic Navigation Systems}},
	author={Amir, Guy and Corsi, Davide and Yerushalmi, Raz and Marzari, Luca and Harel, David and Farinelli, Alessandro and Katz, Guy},
	booktitle={Proc. 29th Int. Conf. on Tools and Algorithms for the Construction and Analysis of Systems (TACAS)},
	pages={607--627},
	year={2023},
}

@inproceedings{AvBlChHeKoPr19,
	title = {{Run-Time Optimization for Learned Controllers through Quantitative Games}},
	author = {Avni, G. and Bloem, R. and Chatterjee, K. and Henzinger, T. and Konighofer, B. and Pranger, S.},
	year = {2019},
	booktitle = {Proc. 31st Int. Conf. on Computer Aided Verification (CAV)},
	pages = {630--649},
}

@inproceedings{AkKeLoPi19,
	title = {{Verification of RNN-Based Neural Agent-Environment Systems}},
	author = {Akintunde, M. and Kevorchian, A. and Lomuscio, A. and Pirovano, E.},
	year = {2019},
	booktitle = {Proc. 33rd AAAI Conf. on Artificial Intelligence (AAAI)},
	pages = {197--210}
}

@inproceedings{BaShShMeSa19,
	Title                    = {{Quantitative Verification of Neural Networks And its Security Applications}},
	Author                   = {Baluta, T. and Shen, S. and Shinde, S. and  Meel, K. and Saxena, P.},
	Booktitle                = {Proc. 26th ACM Conf. on Computer and Communication Security (CCS)},
	Year                     = {2019}
}

@InProceedings{PuTa10,
	Title                    = {{An Abstraction-Refinement Approach to Verification of Artificial Neural Networks}},
	Author                   = {Pulina, L. and Tacchella, A.},
	Booktitle                = {Proc. 22nd Int. Conf. on Computer Aided Verification (CAV)},
	Year                     = {2010},
	Pages                    = {243--257}
}

@inproceedings{SaDuMo19,
	title = {{Reaching Out Towards Fully Verified Autonomous Systems}},
	pages = {22--32},
	booktitle = {Proc. 13th Int. Conf. on Reachability Problems (RP)},
	author = {Sankaranarayanan, S. and Dutta, S. and Mover, S.},
	year = {2019},
}

@inproceedings{ZhShGuGuLeNa20,
	title = {{Verification of Recurrent Neural Networks for Cognitive Tasks via Reachability Analysis}},
	author = {Zhang, H. and Shinn, M. and Gupta, A. and Gurfinkel, A. and Le, N. and Narodytska, N.},
	year = {2020},
	pages = {1690--1697}, 
	booktitle = {Proc. 24th European Conf. on Artificial Intelligence (ECAI)}
}

@inproceedings{AbKe17,
	title={{SMT Solving for Arithmetic Theories: Theory and Tool Support}}, 
	author={Abraham, E. and Kremer, G.},
	booktitle={Proc. 19th Int. Symposium on Symbolic and Numeric Algorithms for Scientific Computing (SYNASC)}, 
	year={2017},
	pages = {1--8},
}

@article{ChDr91,
	title={{Locating Minimal Infeasible Constraint Sets in Linear Programs}},
	author={Chinneck, John and Dravnieks, Erik},
	journal={ORSA Journal on Computing},
	volume={3},
	number={2},
	pages={157--168},
	year={1991},
}

@article{BaBaCoDuKrLaNiNoOzPrReTiZo23,
	title={{Generating and Exploiting Automated Reasoning Proof Certificates}},
	author={Barbosa, Haniel and Barrett, Clark and Cook, Byron and Dutertre, Bruno and Kremer, Gereon and Lachnitt, Hanna and Niemetz, Aina and N{\"o}tzli, Andres and Ozdemir, Alex and Preiner, Mathias and Reynolds, Andrew and Tinelli, Cesare and Zohar, Yoni },
	journal={Communications of the ACM},
	volume={66},
	number={10},
	pages={86--95},
	year={2023},
	publisher={ACM New York, NY, USA}
}

@inproceedings{ZhBrHaZh24,
	title={{Scalable Neural Network Verification with Branch-and-Bound Inferred Cutting Planes}},
	author={Zhou, Duo and Brix, Christopher and Hanasusanto, Grani A and Zhang, Huan},
	booktitle={Proc. 38th Conf. on Neural Information Processing Systems (NeurIPS)},
	year={2024}
}

@inproceedings{SaLa21,
	title={{Reachability Is NP-Complete Even for the Simplest Neural Networks}},
	author={S{\"a}lzer, Marco and Lange, Martin},
	booktitle={Proc. 15th Int. Conf. on Reachability Problems (RP)},
	year={2021},
	pages = {149--164}
}

@inproceedings{ZoBaCsIsJe21,
	title={{Fooling a Complete Neural Network Verifier}},
	author={D{\'a}niel Zombori and Bal{\'a}zs B{\'a}nhelyi and Tibor Csendes and Istv{\'a}n Megyeri and M{\'a}rk Jelasity},
	booktitle={Proc. 9th Int. Conf. on Learning Representations (ICLR)},
	year={2021}
}

@article{HeBi15,
	title={{Proofs for Satisfiability Problems}},
	author={Heule, Marijn JH and Biere, Armin},
	journal={All about Proofs, Proofs for all},
	volume={55},
	number={1},
	pages={1--22},
	year={2015}
}

@inproceedings{DeIsPaStKoKa23,
	title={{Towards a Certified Proof Checker for Deep Neural Network Verification}},
	author={Desmartin, Remi and Isac, Omri and Passmore, Grant and Stark, Kathrin and Komendantskaya, Ekaterina and Katz, Guy},
	booktitle={Proc. 33rd Int. Symposium on Logic-Based Program Synthesis and Transformation (LOPSTR)},
	pages={198--209},
	year={2023}
}

@inproceedings{DeIsKoStPaKa25,
	title={{A Certified Proof Checker for Deep Neural Network Verification in Imandra}},
	author={Remi Desmartin and Omri Isac and Ekaterina Komendantskaya and Kathrin Stark and Grant Passmore and Guy Katz},
	year={2025},
	booktitle={Proc. 16th Int. Conf. on Interactive Theorem Proving
	(ITP)},
	pages={1--21}
}

@inproceedings{SoBi09,
	title={{Minimizing Learned Clauses}},
	author={S{\"o}rensson, Niklas and Biere, Armin},
	booktitle={Proc. 12th Int. Conf. on Theory and Applications of Satisfiability Testing (SAT)},
	pages={237--243},
	year={2009},
	organization={Springer}
}

@article{SiSaMeSi25,
	author = {Singh, Avaljot and Sarita, Yasmin Chandini and Mendis, Charith and Singh, Gagandeep},
	title = {{Automated Verification of Soundness of DNN Certifiers}},
	year = {2025},
	volume = {9},
	journal = {Proc. ACM on Programming Languages}
}

@inproceedings{GuStTr16,
	title={{Minimal Unsatisfiable Core Extraction for SMT}},
	author={Guthmann, Ofer and Strichman, Ofer and Trostanetski, Anna},
	booktitle={In Proc. 16th Int. Conf. on Formal Methods in Computer-Aided Design (FMCAD)},
	pages={57--64},
	year={2016}
}

@article{LiSa08,
	title={{Algorithms for Computing Minimal Unsatisfiable Subsets of Constraint}s},
	author={Liffiton, Mark H and Sakallah, Karem A},
	journal={Journal of Automated Reasoning},
	volume={40},
	pages={1--33},
	year={2008},
	publisher={Springer}
}

@inproceedings{CaArDaIsDiKiRiKo23,
	author    = {Marco Casadio and Luca Arnaboldi and Matthew Daggitt and Omri Isac and Tanvi Dinkar and Daniel Kienitz and Verena Rieser and Ekaterina Komendantskaya},
	title     = {{ANTONIO: Towards a Systematic Method of Generating NLP Benchmarks for Verification}},
	booktitle = {Proc.  6th Workshop on Formal Methods for ML-Enabled Autonomous Systems (FoMLAS)},
	volume    = {16},
	pages     = {59--70},
	year      = {2023},
}

@article{CaDiKoArDaIsKaRiLe25,
	 title={{NLP Verification: Towards a General Methodology for Certifying Robustness}}, journal={European Journal of Applied Mathematics}, author={Casadio, Marco and Dinkar, Tanvi and Komendantskaya, Ekaterina and Arnaboldi, Luca and Daggitt, Matthew L. and Isac, Omri and Katz, Guy and Rieser, Verena and Lemon, Oliver},
	  year={2025},
	  pages={1–58},
}

@misc{KoLaAl24,
	title={{Set-Based Training for Neural Network Verification}},
	author={Koller, Lukas and Ladner, Tobias and Althoff, Matthias},
	Note   = {Technical Report. \url{http://arxiv.org/abs/2401.14961}},
	year={2024}
}

@article{YaHuWeLiLi25,
	title={{Scalable Synthesis of Formally Verified Neural Value Function for Hamilton-Jacobi Reachability Analysis}},
	author={Yang, Yujie and Hu, Hanjiang and Wei, Tianhao and Li, Shengbo Eben and Liu, Changliu},
	journal={Journal of Artificial Intelligence Research},
	volume={83},
	year={2025}
}

@article{DuSiCh22,
	title={{Activation Functions in Deep Learning: A Comprehensive Survey and Benchmark}},
	author={Dubey, S. and Singh, S. and Chaudhuri, B.},
	journal={Neurocomputing},
	year={2022},
}

@misc{VNNCOMP25,
  title        = {{The VNNCOMP 2025 Benchmarks Repository}},
howpublished = {\url{https://github.com/VNN-COMP/vnncomp2025_benchmarks}},
year         = {2025},
key = {VNNCOMP}
}

@article{CoCeAnFeMaRe09,
	title={{Modeling Wine Preferences by Data Mining from Physicochemical Properties}},
	author={Cortez, Paulo and Cerdeira, Ant{\'o}nio and Almeida, Fernando and Matos, Telmo and Reis, Jos{\'e}},
	journal={Decision Support Systems},
	volume={47},
	number={4},
	pages={547--553},
	year={2009}
}

@inproceedings{ChDiHiHaNuRiRuTr18,
	title={{Neural Networks for Safety-Critical Applications --- Challenges, Experiments and Perspectives}},
	author={Cheng, Chih-Hong and Diehl, Frederik and Hinz, Gereon and Hamza, Yassine and N{\"u}hrenberg, Georg and Rickert, Markus and Ruess, Harald and Truong-Le, Michael},
	booktitle={Proc. Int. Conf. on Design, Automation and Test in Europe (DATE)},
	pages={1005--1006},
	year={2018},
	organization={IEEE}
}

@inproceedings{GoPlSeRuSa21,
	title={{Static Analysis of ReLU Neural Networks with Tropical Polyhedra}},
	author={Goubault, Eric and Palumby, S{\'e}bastien and Putot, Sylvie and Rustenholz, Louis and Sankaranarayanan, Sriram},
	booktitle={ Proc. 28th Int. Static Analysis Symposium (SAS)},
	pages={166--190},
	year={2021}
}

@inproceedings{CoDaGiIsJoKaKoLeMaSiWu25,
	title={{Neural Network Verification is a Programming Language Challenge}},
	author={Cordeiro, Lucas C and Daggitt, Matthew L and Girard-Satabin, Julien and Isac, Omri and Johnson, Taylor T and Katz, Guy and Komendantskaya, Ekaterina and Lemesle, Augustin and Manino, Edoardo and {\v{S}}inkarovs, Artjoms and Wu, Haoze},
	booktitle={Proc. 34th European Symposium on Programming (ESOP)},
	pages={206--235},
	year={2025},
}

@misc{ReIs25,
	doi = {10.5281/ZENODO.17169557},
	note={\url{https://zenodo.org/doi/10.5281/zenodo.17169557}},
	author = {Refaeli,  Idan and Isac,  Omri},
	title = {{Proof Minimization in Neural Network Verification (Artifact)}},
	year = {2025}
}


\renewcommand\thesection{\Alph{section}}
\renewcommand\thesubsection{\thesection.\arabic{subsection}}
\setcounter{section}{0}

\end{document}